\documentclass[aps,pra,reprint,showpacs,superscriptaddress]{revtex4-1}
\usepackage{mathtools}
\usepackage{amsmath,amsfonts,amssymb,amsthm,mathrsfs}
\usepackage[T1, OT1]{fontenc}
\DeclareTextSymbolDefault{\DH}{T1}
\DeclareTextSymbolDefault{\DJ}{T1}

\usepackage{booktabs}
\usepackage{float}
\usepackage{siunitx}
\usepackage{xcolor}
\usepackage{tabularx}
\usepackage{graphicx}
\usepackage{subcaption}
\AtBeginDocument{
	\heavyrulewidth=.08em
	\lightrulewidth=.05em
	\cmidrulewidth=.03em
	\belowrulesep=.65ex
	\belowbottomsep=0pt
	\aboverulesep=.4ex
	\abovetopsep=0pt
	\cmidrulesep=\doublerulesep
	\cmidrulekern=.5em
	\defaultaddspace=.5em
}
\usepackage{hyperref}
\usepackage[capitalize]{cleveref}
\hypersetup{ unicode=true, 
	pdftoolbar=true, 
	pdfmenubar=true, 
	pdffitwindow=false, 
	pdfstartview={FitH}, 
	pdftitle={My title}, 
	pdfauthor={Author}, 
	pdfsubject={Subject}, 
	pdfcreator={Creator}, 
	pdfkeywords={keyword1, key2, key3}, 
	pdfnewwindow=true, 
	colorlinks=true, 
	linkcolor=blue, 
	citecolor=blue, 
	filecolor=magenta, 
	urlcolor=blue, 
	pdfencoding=auto, psdextra }
      \usepackage{booktabs}
      \usepackage{tikz}
      \makeatletter
      \tikzset{
         dot diameter/.store in=\dot@diameter,
         dot diameter=1pt,
         dot spacing/.store in=\dot@spacing,
         dot spacing=5pt,
         dots/.style={
               line width=\dot@diameter,
               line cap=round,
               dash pattern=on 0pt off \dot@spacing    
                }
              }
        \makeatother
      \usetikzlibrary{arrows,chains,matrix,positioning,scopes}

\newtheorem{defi}{Definition}

\newtheorem{thm}{Theorem}
\newtheorem{corollary}[thm]{Corollary}

\newtheorem{lem}[thm]{Lemma}
\newtheorem{conjecture}{Conjecture}

\usepackage{algpseudocode} 

\algrenewcommand{\algorithmiccomment}[1]{\hskip2em//\textcolor{green}{\small\tt #1}}

\usepackage{newfloat}
\usepackage{caption}

\DeclareFloatingEnvironment[
fileext=loa,
listname=List of Algorithms,
name=ALGORITHM,
placement=htbp,
]{algorithm}
\DeclareCaptionFormat{algorithms}{\vskip-4pt\rule{\linewidth}{1pt}\par\centering#1#2#3\vskip-5pt\rule{\linewidth}{0.6pt}\vskip-6pt}
\algblock[Input]{Input}{EndInput}
\algblockdefx[Input]{Input}{EndInput}%
[1]{\textbf{Input} #1}%
{}
\newcommand{\real}{\mathbb{R}}
\newcommand{\complex}{\mathbb{C}}

\newcommand{\rank}[1]{\mathrm{rank} (#1)}
\newcommand{\range}[1]{\mathcal{R} (#1)}
\newcommand{\norm}[1]{\left\lVert#1\right\rVert}

\DeclareMathOperator{\trace}{Tr}

\usepackage{mathtools}

\def\symbtype{0} \ifcase\symbtype%

\DeclarePairedDelimiter\bra{\langle}{\rvert}
\DeclarePairedDelimiter\ket{\lvert}{\rangle}
\DeclarePairedDelimiterX\braket[2]{}{}{(#1, #2)}
\DeclarePairedDelimiterX\mket[2]{}{}{#1\otimes#2}

\or

\DeclarePairedDelimiter\bra{\langle}{\rvert}
\DeclarePairedDelimiter\ket{\lvert}{\rangle}
\DeclarePairedDelimiterX\mket[2]{\lvert}{\rangle}{#1,#2}
\DeclarePairedDelimiterX\braket[2]{\langle}{\rangle}{#1 \delimsize\vert#2}

\fi

\usepackage{cancel}


\renewcommand{\intercal}{\mathsf{T}}
\usepackage{dsfont}
\newcommand{\I}{\mathds 1}

\newcommand{\myH}[1][{}]{\mathcal{H}_{#1}}
\newcommand{\etal}{{\sl et~al.}}

\newcommand{\cS}{\mathcal{S}}

\newcommand{\ddk}[1][k]{D_{d,#1}}

\newcommand{\Psym}{P_{\mathrm{sym}}}
\newcommand{\bH}{\myH^{\otimes d}}
\makeatletter
\newcommand{\raisemath}[1]{\mathpalette{\raisem@th{#1}}}
\newcommand{\raisem@th}[3]{\raisebox{#1}{$#2#3$}}
\makeatother
\newcommand{\rankfive}{ rank-$5$ }
\newcommand{\rankfour}{ rank-$4$ }

\begin{document}
\title{Separability of Completely Symmetric States in Multipartite System}
\begin{abstract}
  Symmetry plays an important role in the field of quantum mechanics.  We consider a subclass of
  symmetric quantum states in the multipartite system $N^{\otimes d}$, namely, the completely symmetric states, which
  are invariant under any index permutation. It was conjectured by L. Qian and D. Chu [arXiv:1810.03125 [quant-ph]] that
  the completely symmetric states are separable if and only if it is a convex combination of symmetric pure product
  states. We prove that this conjecture is true for the both bipartite and multipartite cases. Further we
  prove that the completely symmetric state $\rho$ is separable if its rank is at most $5$ or $N+1$. For the states of rank
  $6$ or $N+2$, they are separable if and only if it their range contain a product vector. We apply our results to a few
  widely useful states in quantum information, such as symmetric states, edge states, extreme states and nonnegative
  states. We also study the relation of CS states to Hankel and Toeplitz matrices.
\end{abstract}
 \author{Lin Chen} \email{Electronic adress: linchen@buaa.edu.cn} \affiliation{School
  of Mathematics and Systems Science, Beihang University, Beijing}
\affiliation{International Research Institute for Multidisciplinary Science, Beihang University, Beijing 100191, China}  
 
\author{Delin Chu} \email{Electronic adress: matchudl@nus.edu.sg} \affiliation{Department of Mathematics, National
  University of Singapore, Singapore}
  
\author{Lilong Qian} \email{Electronic adress: qian.lilong@u.nus.edu} \affiliation{Department of Mathematics, National
  University of Singapore, Singapore}
   
\author{Yi Shen} \email{Electronic adress: yishen@buaa.edu.cn} \affiliation{School of Mathematics and Systems Science,
  Beihang University, Beijing} \date{\today}
\maketitle

\tableofcontents
\section{Introduction}
Quantum entanglement, first recognized by Einstein~\cite{einstein1935can} and
Schr\"odinger~\cite{schrodinger1935gegenwartige}, plays a crucial role in the field of quantum computation and quantum
information. It is key resources of  quantum cryptography, quantum teleportation, and quantum key
distribution \cite{nielsen2002quantum}. Therefore, the question of whether a given quantum state is entangled or
separable is of fundamental importance. For a given quantum state $\rho$ acting on the finite-dimensional Hilbert space
$\mathcal{H}_1\otimes \mathcal{H}_{2}$, it is said to be separable if it can be written as a convex linear combination
of pure product quantum states, i.e.,
\begin{equation}
  \label{goesqp:eq:83}
  \rho = \sum_i \lambda_i\ket{x_i,y_i}\bra{x_i,y_i},
\end{equation}
where $\sum_i\lambda_i=1,\lambda_i>0$ and $\ket{x_i}$ and $\ket{y_i}$ are the pure states in  the subspaces
$\mathcal{H}_{1}$ and $\mathcal{H}_2$, respectively.

Despite its wide-importance, to find an efficient method to solve this question in general is considered
to be NP-hard~\cite{gurvits2003classical,gharibian2008strong}. Nevertheless, some subclasses of quantum states have been
studied. For example, the Positive-Partial-Transpose (PPT) states~\cite{Peres1996a} are proved to be separable if
$\mathrm{dim}(\mathcal{H}_{1})\cdot\mathrm{dim}(\mathcal{H}_{2})\leqslant6$~\cite{Horodecki1996,Woronowikz1976}. The
Strong PPT (SPPT) states are considered~\cite{chruscinski2008quantum,chruscinski2008quantum,Qian_2018}. In particular it has been proved
that all the $2\otimes 4$ SPPT states are separable~\cite{ha2013separability}.

The low-rank states are also considered frequently. Kraus~\etal~\cite{kraus2000separability} have proved that the
$2\otimes N$ supported PPT states of rank $N$ is separable. Horodecki~\etal\cite{horodecki2000operational} extend the
results to the general $M\otimes N$ states, that is, all the supported $M\otimes N (M\leqslant N)$ PPT states of rank
$N$ are separable. The results in the general multipartite system are discussed in
Refs.~\cite{fei2003separability,xiao2004ppt}. It tells that the supported
$N_1\otimes N_2\cdots \otimes N_d (N_i\leqslant N_d)$ PPT states of rank $N_d$ is separable if they are
RFRP\@~\cite{Chen_2011red}. Next, it has been proven that multipartite PPT states of rank at most 3 are separable \cite{chen2013separability}. They also proved that the \rankfour PPT states are separable if and
only if their ranges contain at least one product vector. Further, the \rankfour entangled PPT states are necessarily
supported in the two-qutrit or three-qubit subspaces. However, the separability problem regarding \rankfive states remains
unknown. In this paper, we consider completely symmetric (CS) states of rank five. It turns out that these states are separable.

Due to the essential role of symmetry played in the field of quantum entanglement, it is of great importance to explore
the properties of the symmetric states.  A quantum state is called symmetric (bosonic) if it is invariant under the swap
of particles in the symmetric system.    
It is known that PPT symmetric
states of three-qubits are all separable~\cite{eckert2002k}. Moreover, the four-qubit entangled PPT symmetric states
have been characterized in Ref.~\cite{tura2012four}. Further, the existence of entangled PPT symmetric states of five
and six qubits has been reported in Refs.~\cite{toth2009entanglement,toth2010separability}. A well-known subclass of
symmetric states, diagonal symmetric (DS) states, are  proved to be separable if and only if it is
PPT \cite{yu2016separability}.

Within the symmetric
states, the CS states are the main topic in this paper. They are invariant under any index permutation, so they possess a larger symmetry. We conduct a deep investigation into the properties of the CS states and characterization of the
entanglement problem with the low-rank states. We prove that the CS states of rank at most $5$ are separable in
the multipartite system. We further show that the supported CS states of rank at most $N+1$ is separable
in the multipartite system, where $N$ is the dimension of the subsystem. As for the states of rank $6$ and $N+2$, they
are separable if and only if their ranges contains a product vector.  We also construct a $4\otimes 4$ rank $6$
CS states which are entangled. Moreover, all the two-qutrit and multi-qubit CS states are separable. Finally, we extend
some of our results to the general symmetric states and find some applications on edge and extreme states and
nonnegative states. Some subclass of symmetric states is suggested to be considered at the end of the paper.

The remainder of the paper is organized as follows. In Sec.~\ref{sec:pre}, we summarize the preliminaries
on the entanglement and symmetry. In Sec.~\ref{sec:prop-symm-separ}, we present some properties on the CS states. Next, in
Sec.~\ref{sec:rank}, we prove the CS states are separable under some rank conditions. Some applications of our results
are shown in Sec.~\ref{sec:app}. The concluding remarks are given in Sec.~\ref{sec:con}.

\section{Preliminaries}%
\label{sec:pre}
We consider the quantum states in multipartite system $\myH[1]\otimes\myH[2]\cdots\otimes\myH[d]$, where
$\myH[i],i=1,2,\ldots,d$ are Hilbert spaces with dimension $N_d$, respectively.  Mathematically, any quantum state can
be represented by a 
Hermitian semidefinite positive matrix with trace one, which is called the density matrix.
A Hermitian matrix $\rho$ is called semidefinite positive if
\begin{equation}
  \label{eq:positive}
  \bra{x}\rho\ket{x}\geqslant 0 , \forall \ket{x}.
  \end{equation}
The pure state $\rho$ corresponds to the density
matrix whose rank is one, i.e.,
\begin{equation}
  \label{goesqp:eq:84}
  \rho = \ket{\phi}\bra{\phi},
\end{equation}
where $\ket{\phi}$ is a vector in the tensor space $\myH[1]\otimes \myH[2]\cdots\otimes\myH[d]$. In particular, $\rho$
is said to be a pure product state if
\begin{equation}
  \label{goesqp:eq:85}
  \ket{\phi} = \ket{x_1,x_2,\ldots,x_d},
\end{equation}
where $\ket{x_i},i=1,2,\ldots,d$ are vectors in the spaces $\myH[i]$, respectively.
Then a state $\rho$ is called separable if it is a convex combination of the pure product states. Otherwise, it is said
to be entangled.

In the multipartite system, the symmetric (bosonic) space is the subspace which remains invariant under the exchange of
the subsystems. Hence, a quantum state is said to be symmetric if its range is contained in the symmetric space.
In this case,
$\myH[i]=\myH[1],i=2,3,\ldots,d$, denoted by $\myH$.
Particularly, for N-qubit systems, the following Dicke states formulate an orthonormal basis of the symmetric space:
\begin{equation}
  \ket{D_{d,k}} =\frac{1}{\sqrt{C_k^d}} \Psym(\ket{0}^{\otimes k}\ket{1}^{\otimes d-k}),
\end{equation}
where $\Psym$ is the projection onto the symmetric subspace, i.e.,
\[\Psym =\sum_{\pi\in S_d} U_\pi,\]
the sum runs over
all the permutation operators $U_\pi$ of the N-qubit system and $C_k^d = \binom{d}{k}$.

The diagonal symmetric state are defined as
\begin{equation}
  \label{multiSsepv7:eq:1}
  \rho =   \sum_{k=0}^{d}\lambda_k\ket\ddk\bra{\ddk},
\end{equation}
which is a well-known symmetric state being studied frequently. And it is proved to be separable if it is PPT.\@

In Ref.~\cite{lilong2018decomposition}, the authors considered a class of symmetric states, completely symmetric states
(CS states) in the bipartite system. We shall generalize it to the multipartite system.  Here the definition of completely symmetric states is given
formally as follows.

 Suppose $\rho$ is a quantum state in the $d$-partite system
$\bH$ and $\mathrm{dim}(\myH)=N$, then it can be written as
\begin{equation}
  \label{symmetrictensor:eq:1}
  \rho = \sum_{i_1,j_1,\ldots,i_d,j_d=0}^{N-1}\rho_{i_1,j_1,\ldots,i_d,j_d}\ket{i_i,\ldots,i_d}\bra{j_1,\ldots,j_d},
\end{equation}
where $\ket{0},\ket{1},\ldots,\ket{N-1}$ is a natural basis in the space $\myH$.  For simplicity, let
$I_{k} = (i_1,i_2,\ldots,i_{k})$, $J_{k} = (j_1,j_2,\ldots,j_k)$ be the $k$-dimensional multi-index, and
$I = I_d,J = J_d$. Hence $\rho$ can be represented as
\begin{equation}
  \label{eq:rho-rep}
  \rho = \sum_{I,J}\rho_{I,J}\ket{I}\bra{J},
\end{equation}
where
\begin{equation}
  \label{S-sep-v6:eq:2}
  \begin{split}
    \ket{I_{k}} &= \ket{i_1,i_2,\ldots,i_k},\\
    \bra{J_{k}}& = \bra{j_1,j_2,\ldots,j_k}.
  \end{split}
\end{equation}
\begin{defi}
  Let $\rho$ be a quantum state as in \cref{eq:rho-rep}. Then $\rho$ is said to be completely symmetric (CS)
  if
  \begin{equation}
    \label{symmetrictensor:eq:9}
    \rho_{\pi(I,J)} = \rho_{I,J},
  \end{equation}
  for any index permutation $\pi$.
\end{defi}
On the other hand, $\rho$ can be represented as a matrix with an $d$-level nested block
structure~\cite{gurvits2003separable}.
\begin{equation}
  \label{goesqp:eq:90}
  \rho =
  \begin{bmatrix}
    \rho_{00}&\cdots&\rho_{1,N-1}\\
    \vdots&\ddots&\vdots\\
    \rho_{N-1,0}&\cdots&\rho_{N-1,N-1}
  \end{bmatrix},
\end{equation}
where each block is an operator acting on the space $\myH[2]\otimes\cdots\otimes\myH[d]$, i.e.,
\begin{equation}
  \label{goesqp:eq:91}
  \rho_{ij} =\!\!\!\sum_{i_2,j_2,\ldots,i_d,j_d}\!\!\!\rho_{i,j,i_2,j_2,\ldots,i_d,j_d}\ket{i_2,\ldots,i_d}\bra{j_2,\ldots,j_d}.
\end{equation}
Hence we have $\rho_{I,J} = (\rho_{i_k,j_k})_{i_{k+1},j_{k+1},\ldots,i_d,j_d},\forall I,J$. The coefficients
$\rho_{I,J}$ thus correspond to the density matrix $\rho$ in an intuitive way.  Note that, we shall index from
0 instead of 1 in order to be consistent with the standard notations in quantum computation.

%
%
Let us, for the moment, restrict ourselves to supported states, i.e., the states are supported on
$\bH$~\cite{kraus2000separability}.

\newcommand{\redk}[1]{\rho_{\kern-2pt\raisemath{-2pt}{A_{#1}}}}
Denote by $A_1,A_2,\ldots,A_d$ the parties of each subsystem. For the CS states, it
is easy to check that their reduced states are identical, that is,
\begin{equation}
  \label{reduced}
  \redk{k} = \redk{1},\forall k>1,
\end{equation}
where
\begin{equation}
  \label{goesqp:eq:87}
  \begin{split}
    \redk{k}&= (\I_1\otimes \cdots\otimes \I_{k-1}\otimes\trace\otimes\I_{k+1}\otimes\cdots\otimes\I_{d})\rho\\
    &= \sum_{i_l,j_l=0,l\neq k}^{N-1}
    \left(\sum_{i_k=j_k=0}^{N-1}\rho_{I,J}\right)\ket{\phi}\bra{\phi},\\
    \ket{\phi}&= \ket{i_1,\ldots,i_{k-1},i_{k+1},\ldots,i_{d}}.
  \end{split}
\end{equation}
Specifically, denote by $\redk{K}$ the reduced state after applying the trace operator on the $i$-th subsystem, where
$i\in K$ and $K$ is a index subset.  Then $\rho$ is said to be supported on $\bH$ if and only if
$\range{\rho_{\kern-1.5pt\raisemath{-2pt}{A_{\{2:n\}}}}} = \myH$. Note that a CS state $\rho$ is necessarily a positive partial transposed (PPT)
state.

Let us recall the definition of bi-separable and fully separable in the multipartite system.
\begin{defi}
  Suppose $\rho$ is a state in the multipartite system $\myH^{\otimes d}$.
  \begin{enumerate}
    \item   $\rho$ is bi-separable if there exists a bi-partition  $A:B$ for $\{1,2,\ldots,d\}$ such that
          $\rho$ is separable in this bipartite system. 
    \item  $\rho$ is said to be fully separable if it can be written as the convex combination of pure
          product states as in 
          \cref{goesqp:eq:83}. For simplicity, a separable  state will indicate the fully separable one in the
          multipartite case.
  \end{enumerate}
	
\end{defi}

%
Now, we extend the S-seprable states in Ref.~\cite{lilong2018decomposition} to the multipartite system.
For the pure state $\ket{\phi}$ in the multipartite system $\myH^{\otimes d}$, it can be represented by the natural
basis
\begin{equation}
  \label{eq:pure}
  \ket{\phi}=\sum_I a_I\ket{I},a_I\in\real.
\end{equation}
\begin{defi}
  Suppose $\ket\phi$ has the form as \cref{eq:pure}, then it is said to be symmetric if
  \begin{equation}
    a_{\pi(I)} = a_I
  \end{equation}
  for any index permutation $\pi$.
\end{defi}

Note that when the Schmdit rank of real pure state $\ket\phi$ is 1, then $\ket\phi$ can be written as
\begin{equation}
  \ket\phi = \lambda \ket{x}^{\otimes d},\ket{x}\in\real^N
\end{equation}
which is called real symmetric pure product state.

For the mixed state, we have:
\begin{defi}
  Suppose $\rho$ is a quantum state in $\bH$ system. Then it is said to be symmetrically separable (S-separable) if it
  is a convex combination of real symmetric pure product states, i.e.,
  \begin{equation}
    \label{goesqp:eq:63}
    \rho = \sum_{i=0}^{L-1} p_i \ket{x_i}^{\otimes d}\bra{x_i}^{\otimes d};x_i\in\real^N.
  \end{equation}
\end{defi}

According to our results, we present the relation of the symmetric states and its subclasses by Figure 1 and 2. Here DSS and CSS indicate the diagonal symmetric states and completely symmetric states, respectively.
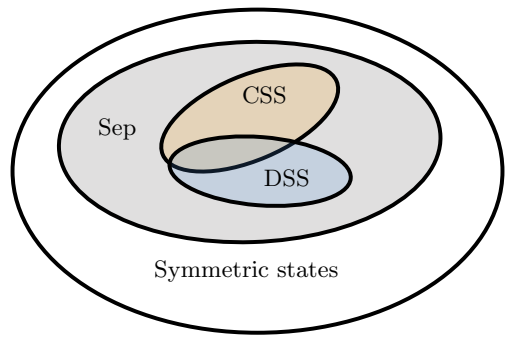
\begin{figure}[H]
  \centering
  \begin{tikzpicture}[scale=0.7,x=0.75pt,y=0.75pt,yscale=-1,xscale=1]
    \draw [line width=1.5] (72.5,177.18) .. controls (72.5,112.84) and (151.52,60.68) .. (249,60.68) .. controls
    (346.48,60.68) and (425.5,112.84) .. (425.5,177.18) .. controls (425.5,241.52) and (346.48,293.68) .. (249,293.68)
    .. controls (151.52,293.68) and (72.5,241.52) .. (72.5,177.18) -- cycle ;
    \draw [fill={rgb, 255:red, 94; green, 90; blue, 88 } ,fill opacity=0.19 ][line width=1.5] (105.94,160.38)
    .. controls (104.72,120.55) and (165.27,86.39) .. (241.17,84.07) .. controls (317.08,81.75) and (379.6,112.16)
    .. (380.81,151.99) .. controls (382.03,191.82) and (321.48,225.98) .. (245.58,228.3) .. controls (169.68,230.62) and
    (107.16,200.21) .. (105.94,160.38) -- cycle ;
    \draw [fill={rgb, 255:red, 245; green, 166; blue, 35 } ,fill opacity=0.2 ][line width=1.5] (180.78,165.37)
    .. controls (174.21,149.87) and (196.91,125.42) .. (231.48,110.77) .. controls (266.05,96.12) and (299.4,96.8)
    .. (305.97,112.3) .. controls (312.54,127.81) and (289.84,152.25) .. (255.27,166.9) .. controls (220.7,181.56) and
    (187.35,180.87) .. (180.78,165.37) -- cycle ;
    \draw [color={rgb, 255:red, 0; green, 0; blue, 0 } ,draw opacity=1 ][fill={rgb, 255:red, 74; green, 144; blue, 226 }
    ,fill opacity=0.18 ][line width=1.5] (185.63,172.01) .. controls (186.7,158.47) and (216.79,149.8)
    .. (252.84,152.66) .. controls (288.89,155.52) and (317.24,168.81) .. (316.16,182.35) .. controls (315.09,195.89)
    and (285,204.55) .. (248.95,201.7) .. controls (212.91,198.84) and (184.56,185.55) .. (185.63,172.01) -- cycle ;
    \draw (148,147) node [align=left] {Sep};
    \draw (270,182) node [align=left] {DSS};
    \draw (254,122) node [align=left] {CSS};
    \draw (241,249) node [align=left] {Symmetric states};
  \end{tikzpicture}
  \caption{Relation of $N$-partite symetric states when $N=2$}
\end{figure}
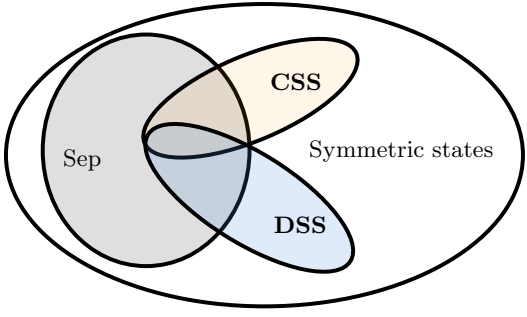
\begin{figure}[H]
  \centering
  \begin{tikzpicture}[scale=0.7,x=0.75pt,y=0.75pt,yscale=-1,xscale=1]

    \draw [line width=1.5] (65,145.03) .. controls (65,84.83) and (148.39,36.03) .. (251.25,36.03) .. controls
    (354.11,36.03) and (437.5,84.83) .. (437.5,145.03) .. controls (437.5,205.23) and (354.11,254.03) .. (251.25,254.03)
    .. controls (148.39,254.03) and (65,205.23) .. (65,145.03) -- cycle ;
    \draw [fill={rgb, 255:red, 94; green, 90; blue, 88 } ,fill opacity=0.19 ][line width=1.5] (91.5,141.53) .. controls
    (91.5,95.42) and (124.85,58.03) .. (166,58.03) .. controls (207.15,58.03) and (240.5,95.42) .. (240.5,141.53)
    .. controls (240.5,187.65) and (207.15,225.03) .. (166,225.03) .. controls (124.85,225.03) and (91.5,187.65)
    .. (91.5,141.53) -- cycle ;
    \draw [fill={rgb, 255:red, 245; green, 166; blue, 35 },fill opacity=0.1 ][line width=1.5] (164.89,136.3) .. controls
    (158.32,120.8) and (187.14,93.76) .. (229.27,75.9) .. controls (271.4,58.05) and (310.88,56.14) .. (317.45,71.64)
    .. controls (324.02,87.14) and (295.19,114.18) .. (253.06,132.04) .. controls (210.93,149.89) and (171.46,151.8)
    .. (164.89,136.3) -- cycle ;
    \draw [color={rgb, 255:red, 0; green, 0; blue, 0 } ,draw opacity=1 ][fill={rgb, 255:red, 74; green, 144; blue, 226 }
    ,fill opacity=0.18 ][line width=1.5] (167.45,130.11) .. controls (176.93,114.99) and (217.23,123.2)
    .. (257.47,148.44) .. controls (297.7,173.68) and (322.63,206.39) .. (313.15,221.51) .. controls (303.67,236.62) and
    (263.37,228.41) .. (223.13,203.17) .. controls (182.9,177.94) and (157.96,145.22) .. (167.45,130.11) -- cycle ;
    \draw (120,149) node [align=left] {Sep};
    \draw (277,195) node [align=left] {\bf DSS};
    \draw (274,92) node [align=left] {\bf CSS};
    \draw (350,142) node [align=left] { Symmetric states};
  \end{tikzpicture}
\caption{Relation of $N$-partite symmetric states when $N\geqslant 3$}
\end{figure}

\section{Properties of CS states}%
\label{sec:prop-symm-separ}
In this section, we investigate the properties of CS and S-separable states.

Recall that in Ref.~\cite{lilong2018decomposition}, the authors conjectured that all the CS states are separable if and
only if
they are S-separable. Now, we extend  this conjecture into the multipartite system.
\begin{conjecture}%
  \label{conj2}
  All the separable CS states in the multipartite system is S-separable.
\end{conjecture}
We will  prove that \cref{conj2} is true for any arbitrary separable CS states at the end of this section. 
To begin with, we show some   properties of the CS states. Due to its special structure, it may admit some
interesting features. One of them is that CS structure will preserve under some transforms.

The following lemma shows that the CS and S-separable states are invariant under the real invertible local operator (RILO)
$A^{\otimes d},A\in \real^{N\times N}$.

\begin{lem}
  Suppose $\rho$ is a quantum state in $\bH$ system and $A$ is a RILO on $\myH$. Then
  \begin{enumerate}
    \item $\rho$ is CS if and only if  $A^{\otimes d}\rho (A^{\otimes d})^{\intercal}$ are CS.
          
    \item $\rho$ is S-separable if and only if  $A^{\otimes d}\rho (A^{\otimes d})^{\intercal}$ are S-separable.    
  \end{enumerate}
\end{lem}
\begin{proof}
  It suffices to prove only one side for the two conclusions. Otherwise, we can consider the state which is applied by
  the RILO ${(A^{-1})}^{\otimes d}$.

  First, we prove CS is invariant under the RILO $A^{\otimes d}$.
 Suppose $\rho$ is a CS states as in \cref{symmetrictensor:eq:1}.
          Let $\sigma = (A^{\otimes d})\rho (A^{\otimes d})^{\intercal}$ and
          \begin{equation}
            \label{goesqp:eq:67}
            A =
            \begin{bmatrix}
              a_{00} & \cdots & a_{0,N-1}\\
              \vdots & \ddots & \vdots\\
              a_{N-1,0} & \cdots & a_{N-1,N-1}
            \end{bmatrix}.
          \end{equation}
          we have,
          \begin{equation}
            \label{goesqp:eq:68}
            A\ket{i} = \sum_{i'=0}^{N-1}a_{i'i}\ket{i'}.
          \end{equation}
          Denote by $I' = (i_1',i_2',\ldots,i_d')$ and $J'= (j_1',j_2',\ldots,j_d')$.
            
          After calculation, we have,
          \begin{equation}
            \label{goesqp:eq:69}
            \begin{split}
              \sigma& = \sum_{I,J,I',J'}\rho_{I,J}
              \prod_{k,l=1}^{d}a_{i_k',i_k}a_{j_l',j_l}\ket{I'}\bra{J'}\\
              & = \sum_{I',J'}\left( \sum_{I,J}\rho_{I,J} \prod_{k,l=1}^{d}a_{i_k',i_k}a_{j_l',j_l}
              \right)\ket{I'}\bra{J'}\\
              & = \sum_{I',J'}\sigma_{I',J'}\ket{I'}\bra{J'},
            \end{split}
          \end{equation}
          where
          \begin{equation}
            \label{S-sep-v6:eq:3}
            \sigma_{I',J'} = \sum_{I,J}\rho_{I,J} \prod_{k,l=1}^{d}a_{i_k',i_k}a_{j_l',j_l}.
          \end{equation}
          In order to prove $\sigma$ is  CS, we only need to prove $\sigma_{I',J'} = \sigma_{\pi(I',J')}$
          where $\pi$ is an arbitrary index  permutation.
          Note that
          \begin{equation}
            \label{multissep:eq:1}
            \begin{split}
              \sigma_{\pi(I',J')} & = \sum_{I,J}\rho_{I,J}
              \prod_{k,l=1}^{d}a_{\pi(i_k'),i_k}a_{\pi(j_l'),j_l},\\
              & = \sum_{I,J}\rho_{I,J}
              \prod_{k,l=1}^{d}a_{\pi(i_k'),\pi(i_k)}a_{\pi(j_l'),\pi(j_l)},\\
              & = \sum_{I,J}\rho_{\pi({I,J})}
              \prod_{k,l=1}^{d}a_{i_k',i_k}a_{j_l',j_l},\\
              & = \sum_{I,J}\rho_{I,J} \prod_{k,l=1}^{d}a_{i_k',i_k}a_{j_l',j_l},\\
              &= \sigma_{I',J'},
            \end{split}
          \end{equation}
          which completes our proof.

          It follows that symmetric separability is also invariant under the operation of $A^{\otimes d}$ by applying $A$ to each subsystem.
\end{proof}

Actually, the CS structure will also hold under some conditions. For example, 
 we have the following useful lemmas.
\begin{lem}
  \label{partition}
  Suppose $\rho$ is a CS state in the multipartite system $\myH^{\otimes d}$.  If $d$ is even, and $A: B$ is a
  bi-partition of the system $\bH$ with $|A|=|B|$. Then $\rho$ is also CS under the bi-partition $A:B$.
\end{lem}
Moreover, the reduced CS states still remain CS.\@
\begin{lem}%
  \label{lem:reduce}
  Suppose $\rho$ is a CS state in the multipartite system $\myH^{\otimes d}$. Then the reduced state $\rho_{A_K}$ is
  CS, where $K$ is a subset of $\{1,2,\ldots,d\}$.
\end{lem}
\begin{proof}
  It suffices to prove for the case $K=\{d\}$, i.e., apply the trace operator to the last
  subsystem.
  Suppose
  \begin{equation}
    \label{multissep:eq:2}
    \rho = \sum_{I,J}\rho_{I,J}\ket{I,J},
  \end{equation}
  where $I = (i_1,i_2,\ldots,i_d)$ and $J = (j_1,j_2,\ldots,j_d)$.  Then
  \begin{equation}
    \label{multissep:eq:3}
    \rho_{A_d} = \sum_{I_{d-1},J_{d-1}}\left( \sum_{k=0}^{N-1}\rho_{I_{d-1},k,J_{d-1},k}
    \right)\ket{I_{d-1}}\bra{J_{k-1}}.
  \end{equation}
  It is easy to check that the coefficients $ \sum_{k=0}^{N-1}\rho_{I_{d-1},k,J_{d-1},k} $ are also invariant under
  index permutation.
\end{proof}

 According to Proposition 13 in
Ref.~\cite{Chen_2013Dim}, the G-invariant states, i.e.\@ invariant under partial transpose, is separable over complex
if and only if it is separable over real. Note that the CS states are necessarily G-invariant, therefore, we have the following lemma, which enables us to consider the separability in
real case.
\begin{lem}%
  \label{real}
  Suppose $\rho$ is a CS state. If $\rho$ is separable over the $\complex$, then it is separable over $\real$, that is,
  $\rho$ can be written as a sum of real pure product states.
\end{lem}

The conception of reducible states is introduced in Ref.~\cite{Chen_2011red}, which enables us to consider the
properties of irreducible components respectively. Recall the definition of reducible states:
\begin{defi}
  \label{def:reducible}
  Suppose $\rho$ is a state in the $d$-partite system $A_1:A_2:\ldots:A_d$. If
  \begin{equation}
    \rho = \sigma + \delta,
  \end{equation}
  where $\sigma$ and $\delta$ are both quantum states in the d-partite system $A_1,A_2,\ldots,A_d$. Then $\rho$ is
  called $A_k$-reducible if
  \begin{equation}
    \range{\sigma_{A_k}}\oplus\range{\delta_{A_k}} = \range{\rho_{A_k}}.
  \end{equation} 
  Otherwise, $\rho$ is said to be $A_K$-irreducible. Specifically, if $\sigma$ and $\delta$ are real operators, we say
  that $\rho$ is $A_K$-irreducible over real.
\end{defi}
It is also proved in Ref.~\cite{Chen_2011red} that the invertible local operator (ILO) transforms the reducible states
to reducible ones. Thus, we can always find an ILO such that $\sigma$ and $\delta$ have orthogonal local
ranges. Moreover, if the two states are real, the ILO can also be chosen as real. It follows that:
\begin{lem}
  If $\rho = \sum_i \rho_i$ is a $A_k$-direct sum, i.e., the ranges of ${(\rho_i)}_{A_K}$ are linearly independent, then
  $\rho$ is separable (PPT) if and only if every $\rho_i$ is separable (PPT).\@
\end{lem}
For the CS states, we also have following result.
\begin{lem}
  If $\rho = \sum_i \rho_i$ is a $A_k$-direct sum and $\rho_i$'s are real, then $\rho$ is CS (S-separable) if and only
  if every $\rho_i$ is CS (S-separable).
\end{lem}
\begin{proof}
  It only needs to prove the necessary part. Suppose $\rho = \sum_i \rho_i$ is a $A_k$-direct sum and every $\rho_i$ is real. First, we prove that
  every $\rho_i$ is CS.\@  By applying a real ILO, we can assume
  \begin{equation}
    \label{eq1}
    \begin{split}
      \range{(\rho_1)_{A_1}} & = \mathrm{span}\{\ket{0},\ldots,\ket{r-1}\},\\
      \range{(\rho_k)_{A_1}} & = \mathrm{span}\{\ket{r},\ldots,\ket{d-1}\},2\leqslant k\leqslant d.
    \end{split}
  \end{equation}
  It suffices to prove
  \begin{equation}
    \range{(\rho_1)_{A_k} }\subset  \mathrm{span}\{\ket{0},\ket{1},\ldots,\ket{r-1}\},2\leqslant k\leqslant d.
  \end{equation}
  Suppose $\rho$ is written in the form of \cref{eq:rho-rep}, Denote by $K$ the set $\{0,1,\ldots,r-1\}$ and $L$ the set
  $\{r,r+1,\ldots,d\}$, then
  \begin{equation}
    \rho_{I,J} =
    \begin{cases}
      {(\rho_1)}_{I,J},& \forall i_1,j_1\in K,\\
      0,& \forall i_1\in K, j_1\in L,\\
      0,& \forall i_1 \in L, j_1\in K.
    \end{cases}
  \end{equation}
  Then for all $I,J \in K^d$, we have
  \begin{equation}
    (\rho_1)_{I,J} = \rho_{I,J}.
  \end{equation}
  Assume $i_1,j_1\in I$ and there is an $i_n \in J ,n\neq 1$. Then
  \begin{equation}
    \begin{split}
      (\rho_1)_{I,J}& = \rho_{I,J}\\
      & = \rho_{i_n,i_{n+1},\ldots,i_d,i_1,\ldots,i_{n-1},j_1,j_2,\ldots,j_d}\\
      & = 0.
    \end{split}
  \end{equation}
  Similar discussion can be applied for $j_n\in J,n\neq 1$. Therefore, $\rho_1$ is the restriction of $\rho$ to the
  subspace $ (\real^r)^{\otimes d}$, hence it is CS.
	
  It is easy to see that $\rho_1$ is S-separable if $\rho$ is S-separable by taking the restriction of $\rho$ to the
  subspace $(\real^r)^{\otimes d}$. This completes our proof.
\end{proof}

It should be noted that $\rho$ is $A_k$-reducible over real if and only if it is $A_n$-reducible ($n\neq k$) over
real. From now on, the CS state $\rho$ is said to be reducible over real if it is $A_k$-reducible
over real for some $k$. 

To prove another feature of CS states, we first introduce the a useful result about the  semidefinite positive  operators.
\begin{lem}%
  \label{lem:positive_operator}
  Suppose $A$ is a semidefinite positive operator, i.e., satisfies the condition~\eqref{eq:positive}. Then
  \begin{enumerate}
    \item $A_{ii}\geqslant 0, \forall i$,
    \item $A_{ij}=A_{ji}=0,\forall j$ if $A_{ii} =0$ for some $i$.
  \end{enumerate}
\end{lem}
With this lemma, we can prove that the CS states are reducible over real if and only if their reduced states are
reducible when $d\geqslant 3$.
\begin{lem}
  \label{eqaul-reduce}
  Suppose $\rho$ is a CS state in the $N^{\otimes d} (d\geqslant 3)$ system. Then $\rho$ is reducible over real if and
  only if $~\trace_k(\rho)$ is reducible over real for some $k$, where $\trace_k$ is the partial trace that applies the
  trace to the $k$-th subsystem.
\end{lem}
\begin{proof}
  Without the loss of generality, we can assume that $k=1$.  It is easy to see that $\rho$ is reducible over real
  implies that of $\trace_1(\rho)$. To prove the inverse direction, suppose $\trace_1(\rho)$ is reducible.
	
  In fact $\trace_1(\rho) = \sum_{i=0}^{N-1}\rho_{ii}$ according to the representation as in \cref{goesqp:eq:90}. Since
  $\rho$ is semidefinite positive, the diagonal blocks $\rho_{ii}$ is semidefinite positive as well. Suppose
  \begin{equation}
    \trace_1(\rho) = \sigma + \delta,
  \end{equation}
  where the local range of $\sigma $ equals $\mathrm{span}\{\ket{0},\ket{1},\ldots,\ket{r-1}\}$ and the local range of
  $\delta$ equals $\mathrm{span}\{\ket{r},\ket{r+1},\ldots,\ket{d-1}\}$. It follows that
  \begin{equation}
    \sum_{i_1=j_1} \rho_{I,J} = 0,\forall i_2<r,i_3\geqslant r,\forall i_3,\ldots,i_d,J.
  \end{equation}
  Note that for a semidefinite positive operator, the diagonal entries are all nonnegative. Let $J = I$, then we have
  \begin{equation}
    \sum_{i_1=0}^{N-1} \rho_{I,I} = 0,\forall i_2<r,i_3\geqslant r,\forall i_3,\ldots,i_d = 0,1,\ldots,d-1,
  \end{equation}
  which implies that
  \begin{equation}
    \rho_{I,I} = 0,\forall i_2<r,i_3\geqslant r,\forall i_1, i_3,i_4,\ldots,i_d = 0,1,\ldots,d-1.
  \end{equation}
  According to \cref{lem:positive_operator}, we have
  \begin{equation}
    \rho_{I,J} = 0,
  \end{equation}
  for $ \forall i_2< r,i_3\geqslant r,\forall i_1, i_3,i_4,\ldots,i_d = 0,1,\ldots,d-1, \forall J$.  Since $\rho$ is CS,
  we have
  \begin{equation}
    \rho_{I,J} = 0, \forall  i_1< r,j_1\geqslant r \text{ or } i_1\geqslant r,j_1<r. 
  \end{equation}
  That is $\rho_{i,j}=0$, if $ i< r, j\geqslant r \text{ or } i\geqslant r, j<r$. Hence $\rho$ is reducible over real
  according to the definition.
\end{proof}

Note that by \cref{lem:reduce}, the reduced CS states remain CS, therefore, $\trace_1(\rho)$ is reducible over real
if and only if $\trace_2(\trace_{1}\rho)$ is reducible over real. Hence we have the following corollary.
\begin{corollary}
  \label{general-reduce}
  Suppose $\rho$ is a CS state in the $A_1:A_2:\ldots:A_d$ system with $d\geqslant 3$. Then $\rho$ is reducible over
  real if and only $\rho_{A_{K}}$ is reducible over real, where $K$ is a subset of $\{1,2,\ldots,d\}$ with
  $|K|\geqslant 2$.
\end{corollary}


It is ready to prove that \cref{conj2} holds for arbitrary multipartite system. First, consider the simplest case where $\rho$ is a pure CS state.
\begin{lem}%
  \label{rank1}
  Any CS pure states are S-separable.
\end{lem}
\begin{proof}
  Note that any quantum state is assumed to be semidefinite positive. Suppose that  $\rho$ is a CS pure state in the
  $\myH^{\otimes d}$ system.  Hence it can be written as
  \begin{equation}
    \label{goesqp:eq:64}
    \rho =  \ket{x_1,\ldots,x_d}\bra{x_1,\ldots,x_d},
  \end{equation}
  where $x_i,i=1,2,\ldots,d$ are unit vectors in $\myH$. If we take the trace to $d-2$ subsystems, then the reduced
  bipartite state is also CS by \cref{lem:reduce}. According to the results for bipartite system in
  Ref.~\cite{lilong2018decomposition}, $x_i$ and $x_j$ must be linearly dependent for any $i\neq j$. Hence, $\rho$ is
  S-separable.
\end{proof}

Next, consider the general bipartite case.
Denote by $\cS$ the subspace spanned by vectors $\{\ket{i,j}+\ket{j,i}\}$.
\begin{lem}
  \label{lem:range}
  Suppose $\rho$ is a separable CS state in the bipartite system $\myH\otimes \myH$. Then
  $\range\rho \subset \mathcal{S}$.
\end{lem}
\begin{proof}
  For any state $\rho$, it can be written as
  \begin{equation}
    \label{add:eq:2}
    \rho = \sum_{ijkl=0}^{N}\rho_{ijkl}\ket{i,k}\bra{jl}.
  \end{equation}
  Since $\rho$ is CS,
  \begin{equation}
    \label{add:eq:3}
    \rho_{ijkl} = \rho_{klij},
  \end{equation}
  which implies that
  \begin{equation}
    \label{add:eq:4}
    \rho =  \sum_{ijkl=0}^{N}\rho_{ijkl}\ket{ki}\bra{jl}.
  \end{equation}
  Add \cref{add:eq:2} and \cref{add:eq:4}, we have
  \begin{equation}
    \label{add:eq:5}
    \begin{split}
      \rho & = \frac{1}{2}\left( \sum_{ijkl=0}^{N-1}\rho_{ijkl} \ket{ik}\bra{jl}+
        \sum_{ijkl=0}^{N-1}\rho_{ijkl} \ket{ki}\bra{jl}\right),\\
      & = \frac{1}{2}\sum_{ijkl=0}^{N-1}\rho_{ijkl}(\ket{ik}+\ket{ki})\bra{jl},
    \end{split}
  \end{equation}
  which completes our proof.
\end{proof}
Note that  the pure states in the space $\cS$ are invariant under the exchange of particles,  we have the following
lemma.
\begin{lem}
  \label{lem:prodinS}
  Suppose $\ket{x,y}$ is a product vector contained in the space $\cS$, then \(\ket{y} = \lambda \ket{x}\) for some
  $\lambda\in \complex$ if $\ket{x}\in\complex^N$.
\end{lem}
If $\ket{x,y}$ is real, then $\lambda$ is real in the above lemma.

The following observation shows that any pure state in $\cS$ can be written as the linear combination of symmetric pure
product states.
\begin{lem}
  \label{le:symbidecomp}
  Suppose $\ket{\psi}$ is a symmetric bipartite pure state. Then

  (i) $\ket{\psi}=\sum_i\ket{a_i,a_i}$ where the $\ket{a_i}$'s are orthogonal states.

  (ii) if $\ket{\psi}$ is real then the $\ket{a_i}$'s in (i) can be chosen as real states up to the imaginary unit
  $i$. That is, $\ket{\psi}=\sum_i(-1)^{c_i}\ket{a_i,a_i}$ where $c_i=0$ or $1$ and $\ket{a_i}$ is real.
\end{lem}
\begin{proof}
  (i) Since $\ket{\psi}$ is a symmetric bipartite pure state, we have
  \begin{eqnarray}
    \ket{\psi}=
    \sum_{i,j}
    m_{i,j}\ket{i,j}	
  \end{eqnarray}  
  with $m_{i,j}=m_{j,i}$ for any $i,j$.  So $M:=(m_{i,j})$ is a symmetric matrix. Using Takagi's factorization we have
  \begin{eqnarray}
    \label{eq:m=udut}
    M=UDU^\intercal	
  \end{eqnarray}
  for a unitary operator $U$ and a diagonal semidefinite positive  operator $D=\mathrm{diag} (d_1,d_2,\cdots)$.  We have
  \begin{align}
    \ket{\psi}
    &= I\otimes M       \sum_i \ket{ii}
    \\\notag\\
    &= I \otimes U D U^\intercal\sum_i \ket{ii}
    \\\notag\\
    &= U \otimes U D \sum_i \ket{ii}
    \\\notag\\
    &= \sum_i d_i 
      U\ket{i}
      \otimes
      U\ket{i}
    \\\notag\\
    &= \sum_i
      \ket{a_i}
      \otimes
      \ket{a_i},
  \end{align}
  where $\ket{a_i}=\sqrt{d_i}U\ket{i}$. So assertion (i) holds.

  (ii) Since $M$ is a real symmetric matrix, using the eigen decomposition, we can choose $U$ as an orthogonal matrix,
  and $D$ as a real diagonal matrix in~Eq.\eqref{eq:m=udut}. Then we repeat the proof of assertion (i). So assertion
  (ii) holds.
\end{proof}

With the above lemma, we can then prove \cref{conj2} is true for arbitrary multipartite system.
\begin{thm}
  \label{thm:conj2}
  \cref{conj2} is true in arbitrary multipartite system.
\end{thm}
\begin{proof}
  Suppose $\rho$ is a separable CS states in the multipartite system $\myH^{\otimes d}$. Then it admits a 
  decomposition
  \begin{equation}
    \rho = \sum_k\lambda_k\ket{\phi_k}\bra{\phi_k},	
  \end{equation}
  where $\ket{\phi_k} $ is a product vector
  \begin{equation}
    \ket{\phi_k} = \ket{x_{k1},x_{k2},\ldots,x_{kd}}.
  \end{equation}
	
  According to \cref{real}, we can assume that $\ket{x_{ki}}\in\real^N$.
	 
  In order to prove that $\rho$ is S-separable, it suffices to prove $\ket{x_{ki}}$ and $\ket{x_{kj}}$ are identical up
  to a real scalar multiplication.

  Take trace to $d-2$ parties except the $i,j$-th subsystem, the reduced state
  \begin{equation}
    \label{multiSsepv5:eq:4}
    \sigma = \sum_i\lambda_i\ket{x_{ki},x_{kj}}\bra{x_{ki},x_{kj}}
  \end{equation}
  is still CS and $\ket{x_{ki},x_{kj}}\in\cS$. By \cref{lem:prodinS},
  \(\ket{x_{kj}} = \lambda_{ki} \ket{x_{ki}}\) for some real constant $\lambda_{ki}$. Therefore, $\rho$ is S-separable,
  which leads to the validness of \cref{conj2}.
\end{proof}

\section{Separability of CS states}
\label{sec:rank}

In this section, we consider the entanglement problem of the CS states. In general, this problem is regarded as
NP-hard. Nevertheless, some results have been obtained with respect to the low-rank states both in the bipartite and
multipartite system. As described in the introduction, any PPT states of rank at most $4$ are separable except in the
$3\otimes 3$ and $2\otimes 2 \otimes 2$ case. For the certain two cases, they are separable if and only if they have a
product vector in their range. It is also proved that any supported $M\otimes N$ states of rank $N$ are
separable. However, in the multipartite system, the supported
$N_1\otimes N_2\cdots \otimes N_d (N_1\leqslant N_2\ldots\leqslant N_d)$ state $\rho$ of rank $N_d$ is separable only if
there exists a product vector $\ket{\phi}=\ket{x_1,x_2,\ldots,x_{d-1}}$ such that
\begin{equation}
  \label{multiSsepv5:eq:17}
  \rank{\bra{\phi}\rho\ket{\phi}}=N_d.
\end{equation}

In this part, we will develop stronger results about the entanglement problem of CS states. In addition, all the separable
CS states are necessarily S-separable.
First, for the CS states,  we prove that the full separability is equivalent to bi-separability.

Let $\pi_0$ be the periodic index permutation
\begin{equation}
  \pi_0 (1,2,\ldots,d)\Rightarrow (2,3,\ldots,d,1).
\end{equation} 
Let $\mathcal{F}$ be the subspace spanned by the vectors which is invariant under the periodic permutation:
\begin{equation}
  \label{eq:F}
  \mathcal{F} = \{  \phi\in\real^N:\pi_0^k(\ket{\phi}) =\ket{\phi},\forall k>0 \},
\end{equation}
where
\begin{equation}
  \pi_0^k = \underset{k}{\underbrace{\pi_0\circ\pi_0\circ\cdots\circ\pi_0}}.
\end{equation}

\begin{lem}
  \label{lem:multi-range}
  The range of CS states is contained in  $\mathcal{F}$.
\end{lem}
\begin{proof}
  Suppose $\rho$ is a CS states as in \cref{eq:rho-rep}. Note that $\rho$ is CS, hence
  \begin{equation}
    \rho_{\pi_o^k(I),J}=\rho_{I,J},\forall k>0.
  \end{equation}
  Therefore,
  \begin{equation}
    \begin{split}
      \rho & = \frac{1}{d} \sum_{k=0}^{d-1}\left(\sum_{I,J}\rho_{\pi_0^k(I),J}\ket{I}\bra{J}\right),\\
      & = \frac{1}{d} \sum_{I,J}\rho_{I,J} \left(\sum_{k=0}^{d-1}\pi_0^k(\ket{I})\right) (J).
    \end{split}
  \end{equation}
  Since  $\sum_{k=0}^{d-1}\pi_0^k(\ket{I})$ is contained in $\mathcal{F}$, it follows that 
  $\range\rho\subset \mathcal{F}$.
\end{proof}

\begin{thm}
  \label{thm:bi=ful}
  The  CS states are fully separable if and only if they are  bi-separable.
\end{thm}
\begin{proof}
  Suppose that $\rho$ is a CS state in the mulripartite system $\myH^{\otimes d}$. It only needs to prove that the
  bi-separability implies the full separability. Assume $\rho$ is bi-separable under the bipartition
  $1,\ldots,k:k+1,\ldots,d$.
	
  According to \cref{real}, $\rho$ can be written as
  \begin{equation}
    \begin{split}
      \rho& =\sum_i\lambda_i \ket{\varphi_i}\bra{\varphi_i},\\
      &= \sum_i \lambda_i\ket{\phi_i,\psi_i}\bra{\phi_i,\psi_i},
    \end{split}
  \end{equation}
  where $\lambda_i>0, \ket{\varphi_i} = \ket{\phi_i,\psi_i},\ket{\phi_i}\in\real^{N^k},\ket{\psi_i}\in\real^{N^{d-k}} $.
	
  By \cref{lem:multi-range}, each  $\ket{\varphi_i}$ is a symmetric pure state, which is invariant under the
  periodic permutation $\pi_0$.  Hence $\ket{\varphi_i}$ should also be separable under the bipartition
  $\{2,3,\ldots,k+1\}:\{k+2,\ldots,d,1\}$.
	
  Take the trace to the $k+1,k+2,\ldots,d$-th parties, then $\ket{\phi_i}$ is a separable pure product vector under the
  partition $\{2,3,\ldots,k\}:1$. That is, $\ket{\phi_i}$ can be written as
  \begin{equation}
    \ket{\phi_i} = \ket{x_i,\phi_i^{(1)}},
  \end{equation}
  where  $\ket{x_i}\in\real^N$ and $\ket{\phi_i^{(1}}\in\real^{N^{k-1}}$.
	
  Apply  the similar discussion  to $\trace_1(\rho)$, we can conclude that $\varphi$
  is fully separable, which completes our proof.
\end{proof}
With \cref{thm:bi=ful} and \cref{thm:conj2},  any bi-separable CS states are necessarily S-separable.  
The following lemma describes the product vectors in the range or kernel of CS states..
\begin{lem}
  \label{le:prod}
  Suppose $\rho$ is a bipartite CS state, the product vectors $\ket{a,a}\in\range{\rho}$ and
  $\ket{b,c}\in\ker\rho$. Then (i) $\ket{a^*,a^*}\in\range{\rho}$; (ii)
  $\ket{b^*,c},\ket{b,c^*},\ket{b^*,c^*}\in\ker\rho$.
\end{lem}
\begin{proof}
  (i) The assertion follows from the fact that $\rho$ is real.

  (ii) Note that $\rho$ is unchanged under the partial transpose operators, hence, we have
  \begin{eqnarray}
    0&&=\bra{b,c}\rho\ket{b,c}	
        =\bra{b^*,c}\rho\ket{b^*,c}\notag
    \\
     &&=\bra{b,c^*}\rho\ket{b,c^*}
        =\bra{b^*,c^*}\rho\ket{b^*,c^*}.
  \end{eqnarray}	
  So assertion (ii) holds.
\end{proof}

\subsection{Multi-qubit state}
\begin{thm}
  \label{thm:dqubit}
  All the multi-qubit CS states are separable.
\end{thm}
\begin{proof}
  In order to prove this result, we claim that any $2\otimes N$ G-invariant quantum states are separable.
	
  Suppose $\rho$ is a G-invariant quantum state supported in the system $2\otimes N$. Since $\rho$ is semidefinite positive,
  $\rank{\rho}\geqslant N$.
	
  If $\rank{\rho}=N$, it is separable according to the result in Ref.\cite{kraus2000separability} since G-invariant
  states are necessarily PPT.\@
	
  Assume  $\rank{\rho}>N$. Then any G-invariant states are separable if
  for any G-invariant states of rank at least $N+1$ contains a real product vector in its range. 
	
  Suppose the kernel of $\rho$ is spanned by the vectors $\ket{f_i} = \ket{0,f_{i0}}+\ket{1,f_{i1}}$. Let
  \begin{equation}
    \ket{x} = \begin{bmatrix}
      x_0\\x_1
    \end{bmatrix},x_0,x_1\in\real.
  \end{equation}
	
  The range of $\rho$ contains a real product vector if and only if the following equations has a real solution:
  \begin{equation}
    \label{eq:realproduct2N}
    (x_0 F_0 +x_1F_1)y = 0,
  \end{equation}
  where
  \begin{equation}
    \begin{split}
      F_1 &= \begin{bmatrix}
	\bra{f_{11}}\\
	\bra{f_{21}}\\
	\vdots\\
	\bra{f_{r1}}
      \end{bmatrix},\\
      F_2 &= \begin{bmatrix}
	\bra{f_{11}}\\
	\bra{f_{21}}\\
	\vdots\\
	\bra{f_{r1}}
      \end{bmatrix},\\
      r& = 2N-\rank{\rho}\leqslant N-1.
    \end{split}
  \end{equation}
  Since $r<N$, the above linear system~\eqref{eq:realproduct2N} has a real solution for any real number $x_1,x_2$.
	
  Suppose $\ket{x_1,y_1}$ is a real product vector in the range of $\rho$, there exists a positive number $\lambda_1$
  such that
  \begin{equation}
    \rho^{(1)}=\rho - \lambda_1\ket{x_1,y_1}\bra{x_1,y_1}>0,
  \end{equation}
  and $\rank{\rho^{(1)}} = \rank{\rho} -1$. By the deduction on the ranks, we can then conclude $\rho$ is separable.
	
  Note that an arbitrary $d$-qubit CS state under the bi-partition $1:\{2,3,\ldots,d\}$ is G-invariant, hence
  bi-separable. According to \cref{thm:bi=ful}, it is separable.
\end{proof}

\subsection{Two-qutrit state}
      
Using the Segre variety, we have the following lemma for bipartite system.
\begin{lem}
  \label{lem:segre}
  Suppose  $ n\geqslant \frac{N(N-1)}{2}+1$, then any $n$-dimensional subspace of $\mathcal{S}$ contains a complex product
  vector.
\end{lem}

\begin{thm}
  \label{rank5-bi}
  In the bipartite system, any CS states are separable if $N\leqslant 3$.
\end{thm}
\begin{proof}
  It has been proved for the case $N\leqslant 2$ in \cref{thm:dqubit}. We need only consider the two-qutrit system.
  Suppose $\rho$ is a $3\otimes 3$ CS state. If $\rank{\rho}\leqslant 3$, then it is separable according to the results
  in Refs.~\cite{Chen_2011red}.  If $\rank{\rho} =4$, then $\rho$ is separable if and only if its range contains a product
  vector. It follows from~ \cref{lem:segre} that $\rho$ is separable.

  Assume $\rank{\rho}\geqslant 5$, we claim that its range must contains a real pure product vector $\ket{x,x}$.  In
  fact, this claim will lead to the separability of $\rho$ of any rank.

  Denote by $\ket{\varphi_i},i=0,1,2,3,4$ the 5 linearly independent vectors which are contained in the range of
  $\rho$. On the other hand, the $\range{\rho}$ can be represented by the following vectors:
  \begin{equation}
    \begin{split}
      \ket{\phi_0} & = \ket{0,0},\\
      \ket{\phi_1} & = \ket{1,1},\\
      \ket{\phi_2 }& = \ket{2,2},\\
      \ket{\phi_3 }& = \ket{0,1}+\ket{1,0},\\
      \ket{\phi_4} & = \ket{0,2}+\ket{2,0},\\
      \ket{\phi_5} & = \ket{1,2}+\ket{2,1}.
    \end{split}
  \end{equation}
  Hence, we have
  \begin{equation}
    \ket{\varphi_{i}} = \sum_{j=0}^{5} a_{ij}\ket{\phi_j},i=0,1,\ldots,4,a_{ij}\in\real.
  \end{equation}
  Since $\ket{\varphi_i},i=0,\ldots,4$,$\ket{\phi_0}$, and $\ket{\phi_1}$ are linearly dependent, there exists a linear
  combination such that 
  \begin{equation}
    \sum_{i=0}^{5}b_i\ket{\varphi_i} = c_0\ket{\phi_0} + d_0\ket{\phi_1}\in\range{\rho},b_i\in\real.
  \end{equation}
  If $c_0d_0=0$, there we can conclude that $\range{\rho}$ contains a real product vector, which completes our
  proof. Otherwise, we can assume $c_0 = 1$. Denote by $\ket{\psi_0} = \ket{\phi_0}+d_0\ket{\phi_1}$.  Similarly, we
  have the following vectors in the range of $\rho$,
  \begin{equation}
    \begin{split}
      \ket{\psi_1}&=\ket{\phi_0} + d_1\ket{\phi_2},\\
      \ket{\psi_2}& =\ket{\phi_1}+ d_2\ket{\phi_2},\\
      \ket{\psi_3}& =\ket{\phi_3} +e_0\ket{\phi_1},\\
      \ket{\psi_4}&= \ket{\phi_5} + e_1\ket{\phi_1},\\
      \ket{\psi_5}& =\ket{\phi_4} + e_2\ket{\phi_2}.
    \end{split}
  \end{equation}
  We can assume that at least one of $d_0,d_1,d_2$ is greater than 0. Otherwise
  $\psi_0 - d_0\psi_2 = \ket{0,0} - d_0d_2\ket{2,2}$ is a vector in $\range{\rho}$ which satisfies the requirement.

  For simplicity,  assume $d_0>0$. Then the two vectors
  \begin{equation}
    \label{multiSsepv1:eq:1}
    \begin{split}
      \psi_0&=\ket{0,0}  + d_0\ket{1,1},
      \psi_3 =\ket{0,1}+\ket{1,0} + e_0\ket{1,1},
    \end{split}
  \end{equation}
  are in the range of $\rho$. Note that by the Schmidt rank of pure states, the vector $x\ket{\psi_0}+\ket{\psi_3}$ is a
  produce vector if and only if
  \begin{equation}
    \label{multiSsepv1:eq:2}
    \mathrm{det}
    \begin{bmatrix}
      x&1\\
      1&e_0x+d_0
    \end{bmatrix}
    = e_0x^2 +d_0x -1 =0.
  \end{equation}
  Since $e_0>0$, there exists a real solution, denoted by $x_0$, of the above equation. Hence
  $x_0\ket{\psi_0}+\ket{\psi_3}$ is a real product vector which is contained in the range of $\rho$, which completes our
  proof.
\end{proof}
\subsection{$\rank{\rho}=N$}
\begin{thm}
  \label{rankN}
  Suppose $\rho$ is a CS state supported in the multipartite system $N^{\otimes d}$. If $\rank{\rho}=N$, then $\rho$ is
  separable.
\end{thm}
\begin{proof}
  Consider the bi-partition $1:\{2,\ldots,d\}$, then $\rho$ is supported in the system $N\otimes R, R\geqslant
  N$. According to the result in Ref.~\cite{horodecki2000operational}, it is separable. By \cref{thm:bi=ful}, it is
  fully separable, hence S-separable.
\end{proof}
For the CS states, the condition of separability is weaker than what presented in Ref.~\cite{wang2005canonical}.

\subsection{ $\rank{\rho}=N+1$}
\newcommand{\opb}{\oplus_{\kern-1.5pt\raisemath{-2pt}{B}}}
Let us recall a useful result in Ref.~\cite{Chen_2013properties}. Denote by $A,B$ the parties for each subsystem. We
shall refer to $\rho\opb\sigma$ in the sense that the range of $\rho_B$ and $\sigma_B$ have 
no intersection except the zero vector.
\begin{lem}
  \label{Chen:N+1}
  Suppose $\rho$ is an $M\otimes N (2<M\leqslant N) $ PPT state of rank $N+1$, then $\rho$ is either a sum of $N+1$ pure
  products states or 
  \begin{equation}
    \rho = \rho_1\opb\rho_2\opb\cdots\opb\rho_k\opb\sigma,
  \end{equation} 
  where $\rho_i$'s are pure product states and $\sigma$ is a bipartite PPT state of rank $N+1-k$.
\end{lem}

\begin{lem}
  \label{lem:rankN+1-bi}
  Any $N\otimes N$ CS states of rank $N+1$ are separable.
\end{lem}
\begin{proof}
  Suppose $\rho$ is an $N\otimes N$ CS state of rank $N+1$, which is necessarily  PPT.\@ To prove the assertion, it suffices to
  assume $N>2$. By \cref{Chen:N+1}, we obtain that
  \begin{eqnarray}
    \rho=
    \rho_1\opb\cdots
    \opb\rho_k
    \opb\sigma,	
  \end{eqnarray}
  where $\rho_i$'s are pure product states.

  Since $\rho$ is CS, one can show that every $\rho_i$ is real symmetric pure product state. So $\sigma$ is CS as well.
                
  Hence, $\sigma$ is a bipartite state of $\rank{\sigma_B} =N-k$ and $\rank\sigma=N-k+1$.  Since $\rho$ is CS, it is
  symmetric. So $\rank{\sigma_A}=N-k$.
		
  We prove the assertion by contradiction. If $\rho$ is entangled, Theorem 45 in Ref.~\cite{Chen_2013properties}  says
  that $\rank{\sigma_A}=2$ or $3$. It means that $\rank\sigma\le4$.  We have a contradiction with the proven fact 
  that every CS state of rank at most four is separable. This completes the proof.
\end{proof}

\begin{thm}
  \label{thm:mul-n+1}
  Any $N^{\otimes d}$ CS states of rank $N+1$ are separable.
\end{thm}
\begin{proof}
  Consider the bi-partition $A_1:A_2,\ldots,A_d$. We have that
  \begin{equation}
    N\leqslant \rank{\trace_1\rho} \leqslant N+1.
  \end{equation}
  If $\rank{\trace_1\rho} =N+1$, then $\rho$ is bi-separable by \cref{Chen:N+1}, which is separable according to
  \cref{thm:bi=ful}.
	
  If $ \rank{\trace_1\rho} =N,N>3$, by \cref{Chen:N+1}, we can find a reducible components $\rho_1$ such that
  \begin{equation}
    \rho = \rho_1 + \sigma,
  \end{equation}
  where $\rho_1$ is a CS pure product state and $\sigma $ is of rank $N$ CS state supported in the ${(N-1)}^{\otimes d}$
  space. Apply the similar discussion recursively on $\sigma$, hence we can assume $\sigma$ is a  $3^{\otimes d}$ \rankfour
  CS state, 
  which is separable  according to \cref{rank4} in the next subsection.
	
  If $ \rank{\trace_1\rho} =N,N\leqslant 2$, it is separable by \cref{thm:dqubit}.
	
  Above all,  any $N^{\otimes d}$ CS states of rank $N+1$ are separable.
\end{proof}
A corollary about the rank-$(N+2)$ CS states follows from \cref{thm:mul-n+1}.
\begin{corollary}
  Suppose $\rho$ is a CS state supported in the multipartite system $N^{\otimes d}$. If $\rank{\rho} = N+2$, then it is
  separable if and only if $\range{\rho}$ contains a real product vector.
\end{corollary}
\subsection{ $\rank{\rho}\leqslant 5$}
\label{subsec:rank5}
Before all else, we prove all the rank-$4$ states are separable.
\begin{thm}
  \label{rank4}
  Any CS states of rank at most $4$ in the multipartite system are separable.
\end{thm}
\begin{proof}
  First, we consider the bipartite system.
	
  Suppose $\rho$ is a CS state in the bipartite system $\myH\otimes \myH$.  In order to prove that $\rho$ is
  S-separable, by \cref{thm:conj2}, it only need to prove $\rho$ is separable.
  
  Note that a CS state is necessarily PPT.\@  In addition, it was proved that the PPT states of rank at most 3 are
  separable in Refs.~\cite{Chen_2011red}. 

  It remains to prove it when the rank is 4. According to the result in Ref.~\cite{Chen_2011red}, only the $3\otimes 3$
  state could be entangled and it is separable if and only if its range contains a product vector.

  By \cref{lem:segre}, the \rankfour states in $3\otimes 3$ must contains a product vector. Therefore, $\rho$ is
  separable.
  
  Finally, we consider the general multipartite system.
  
  By the results in Ref.~\cite{chen2013separability}, all the PPT states of rank at most $3$ are separable. And the rank
  $4$ PPT states are separable except in the $3\otimes 3$ and $2\otimes 2\otimes 2$ cases. It thus only to prove that
  the $2\otimes 2\otimes 2$ \rankfour state is separable. It follows by \cref{thm:dqubit} that our proof completes.
\end{proof}

Furthermore, we can prove that any rank-$5$ CS states are separable.
\begin{thm}
  \label{rabk5}
  Any CS states of rank at most $5$ are separable in the multipartite system.
\end{thm}
\begin{proof}
  By \cref{rank4}, it suffices to prove the \rankfive~CS states are separable.
	
  Suppose $\rho$ is a rank-$5$ CS state supported in the system $N^{\otimes d}$. It is easy to show that $N\leqslant 5$.
	
  If $N=5$, by \cref{rankN}, $\rho$ is S-separable.
	
  If $N=4$, consider the bi-partition $A_1:A_2,\ldots,A_d$. Note that $4\leqslant \rank{\trace_1 \rho}\leqslant 5$. If
  $\rank{\trace_1 \rho}=5$, $\rho$ is bi-separable under this partition.  If $\rank{\trace_1 \rho}=4$, $\rho$ is a
  $4\otimes 4$ \rankfive state under this partition. By \cref{Chen:N+1}, there exist a real product vector such
  $\ket{x,y}\in\range{\rho}$, where $\ket{x}\in \myH$ and $\ket{y}\in\otimes^{d-1}\myH$. Since
  $\ket{x,y}\in\mathcal{F}$, where $\mathcal{F}$ is periodic symmetric subspace as defined in \cref{eq:F}, $\ket{x,y}$
  can be written as $\gamma \ket{x}^{\otimes d}$. Therefore, there exists a positive constant $\lambda$ such that
  $\rho - \lambda \ket{x}^{\otimes d}\bra{x}^{\otimes d}$ is CS, semidefinite positive and of rank four, which is separable by
  \cref{rank4}.  Hence $\rho$ is separable.
	
  If $N=3$, consider the bi-partition $A_1:A_2,\ldots,A_d$.  Note that $3\leqslant \rank{\trace_1 \rho}\leqslant 5$.
  $ \rank{\trace_1 \rho}=5$, $\rho$ is bi-separable by \cref{rankN} and hence fully separable.
   
  If $ \rank{\trace_1 \rho}=4$,  $\trace_1\rho$ is S-separable by \cref{rank4}. Moreover, it is a sum $4$ CS pure
  product states:
  \begin{equation}
    \label{multiSsepv2:eq:2}
    \trace_{1}\rho = \sum_{i=0}^{3}\ket{x_i}^{\otimes d-1}\bra{x_i}^{\otimes d-1},
  \end{equation}
  where $\ket{x_i}\in\real^3$. Since $\trace_{1,2}\rho$ is of rank at least $3$, we can assume $\ket{x_i},i=0,1,2$ are
  linearly independent. If $\ket{x_3}$ is a linearly combination of two vectors of $\ket{x_0},\ket{x_1},\ket{x_2}$, then
  it is reducible over 
  real, which implies that $\rho$ is also reducible over real by \cref{general-reduce}. Since all the irreducible
  components are still CS states of smaller ranks, they are separable according to \cref{rank4}. Hence $\rho$ is
  separable. Therefore, we can further assume that any
  three vectors of $\ket{x_i}$ are linearly independent. It implies that there exists  a real ILO such that
  \begin{equation}
    \label{multiSsepv2:eq:1}
    \begin{split}
      &(A^{\otimes (d-1)})(\trace_{1}\rho)(A^{\intercal})^{\otimes (d-1)}) \\
      &=
      \sum_{i=0}^2\gamma_i{(\ket{i}\bra{i})}^{\otimes d-1}+ \gamma_3{(\ket{x}\bra{x})}^{\otimes d-1},
    \end{split}
  \end{equation}
  where $\ket{x} = \ket{0}+\ket{1}+\ket{2}$ and $\gamma_i>0,i=0,1,2,3$.

  For simplicity, we substitute $\rho$ with $( A^{\otimes d})\rho(A^{\intercal})^{\otimes d}$ but still denote this
  state by $\rho$. Suppose $\ket\phi$ is a vector in the range of $\rho$, then it must has the form
  \begin{equation}
    \label{multiSsepv2:eq:4}
    \begin{split}
      \ket\phi & = \sum_{i=0}^{2}\ket{a_i}\ket{i}^{\otimes d-1} + \ket{a_3}\ket{x}^{\otimes d-1}.
    \end{split}
  \end{equation}
  Since $\ket\phi$ is in the space $\mathcal{F}$, the periodic subspace, we have
  \begin{equation}
    \label{multiSsepv2:eq:5}
    \begin{split}
      &\sum_{i=0}^{2}\ket{a_i}\ket{i}^{\otimes d-1} + \ket{a_3}\ket{x}^{\otimes d-1}\\
      & = \sum_{i=0}^{2}\ket{i,a_i}\ket{i}^{\otimes d-2} + \ket{x,a_3}\ket{x}^{\otimes d-2}.
    \end{split}
  \end{equation}
  By solving the above  system of equations, we have
  \begin{equation}
    \label{multiSsepv2:eq:6}
    \ket{a_i}\propto \ket{i},\ket{a_3}\propto\ket{x}.
  \end{equation}
  That is, $\ket\phi$ has the following form
  \begin{equation}
    \label{multiSsepv2:eq:7}
    \ket\phi = \sum_{i=0}^2\lambda_i\ket{i}^{\otimes d} + \lambda_3\ket{x}^{\otimes d},
  \end{equation}
  where $\lambda_i\in\real$. However, the range spanned by these vectors is rank at most 4, which leads to a
  contradiction.

  If $\rank{\trace_1 \rho}=3$, then the reduced states is a $3\otimes 3$ rank 3 states, which is a sum of three real
  pure product states according to the results in Ref.~\cite{horodecki2000operational}. It follows that $\trace_1\rho$
  is reducible over real. By \cref{general-reduce}, $\rho$ is also reducible over real. Note that for all the
  irreducible components, they are CS and of rank at most 4, hence separable.  It has been prove that for $N=2$, the CS
  state is S-separable and it is trivial case when $N=1$. Therefore, $\rho$ is separable in this case.
   
  Above all, any multipartite CS states of rank at most $5$ are separable.
\end{proof}

\subsection{ $\rank{\rho}$= 6}
We have proved that many Cs states of low ranks are separable. Unfortunately, not all CS states are separable although
they have a special structure. Here, we construct a rank-$6$ entangled CS states as follows. Let
\begin{equation}
  \label{eq:sigma}
  \sigma = \sum_{i=0}^6\lambda_{i}\ket{\phi_i}\bra{\phi_i},\lambda_i>0,
\end{equation}
where
\begin{equation}
  \label{multiSsepv5:eq:2}
  \ket{\phi_i} = \ket{x_i,x_i},\text { for } i=0,1,\ldots, 6,
\end{equation}
and
\begin{equation}
  \label{multiSsepv5:eq:3}
  \begin{split}
    \ket{x_i}& = \ket{i},i=0,1,2,3\\
    \ket{x_4}&= \ket{0}+\ket{1}+\ket{2}+\ket{3},\\
    \ket{x_5}&= \ket{0}+2\ket{1}+3\ket{2}+4\ket{3},\\
    \ket{x_6}&= \ket{0}-2\ket{1}+3\ket{2}-4\ket{3}.\\
  \end{split}
\end{equation}
We claim that  that the range of $\sigma$  contains exactly$8$ real symmetric product vectors, that is
$\ket{\phi_i},i=0,1,\ldots,6$ and $\ket{\phi_7} = \ket{x_7,x_7}$, where
\begin{equation}
  \label{multiSsepv5:eq:5}
  \ket{x_7} = \ket{0} - \frac{8}{3}\ket{1}+\ket{2}-\frac{8}{3}\ket{3}.
\end{equation}
In order to be concise, we append the proof in \cref{append}.

Note that there exists a positive constant $\lambda$ such that
\begin{eqnarray}
  \label{eq:4x4entangled}	
  \rho = \sigma-\lambda \ket{\phi_7}\bra{\phi_7} \geqslant 0
\end{eqnarray}
has rank $6$. We claim that $\rho$ is entangled. Otherwise, assume $\rho$ is separable, and thus
S-separable. Since $\range{\rho}\subset\range{\sigma}$, the symmetric product vector contained in $\rho$ must only be
$\ket{\phi_0},\ldots,\ket{\phi_{6}}$. Note that the rank of $\rho$ is $6$, without the loss of generality, we can hence
assume that 
$\rho = \sum_{i=0}^{5}\gamma_i\ket{\phi_i}\bra{\phi_i}$.

Therefore
$\sigma =\sum_{i=0}^{5}\gamma_i\ket{\phi_i}\bra{\phi_i}+ \lambda \ket{\phi_7}\bra{\phi_7}$.  By \cref{eq:sigma},
we have,
\begin{equation}
  \label{multiSsepv5:eq:6}
  \sum_{i=0}^5(\gamma_i-\lambda_i)\ket{\phi_i}\bra{\phi_i}+ \lambda_6\ket{\phi_6}\bra{\phi_6}+  \lambda\ket{\phi_7}\bra{\phi_7}=0.
\end{equation}
Note that the operators $\ket{\phi_i}\bra{\phi_i},i=0,1,\ldots,7$ are linearly independent, hence we have $\lambda_6=0$,
which is a 
contradiction. Therefore, $\rho$ must be entangled.

In fact, for rank $6$ bipartite CS states, we have the following sufficient and necessary condition:
\begin{lem}
  \label{le:lastelm}
  Suppose $\rho$ is CS state of rank six in the bipartite system. Then $\rho$ is separable if and only if the range of
  $\rho$ has a product state.
\end{lem}
\begin{proof}
  Suppose $\ket{x,x}\in\range{\rho}$. The assertion holds if $\ket{x}$ is real. We assume that $\ket{x}$ is not
  proportional to a real vector. Up to a real ILO,  we may assume that $\ket{x}=\ket{0}+i\ket{1}$,
  and $\rho$ is still a CS state of rank six. Hence $\ket{01}+\ket{10},\ket{00}-\ket{11}\in\range{\rho}$. Let
  \begin{eqnarray}
    \label{eq:h}
    \begin{split}
      H&=(\ket{01}+\ket{10})(\bra{01}+\bra{10})\\
      &-(\ket{00}-\ket{11})(\bra{00}-\bra{11})
    \end{split}
  \end{eqnarray}
  be a Hermitian matrix. One can verify that $H$ is a CS matrix. So there exists a small enough $\epsilon>0$ such that
  $\sigma=\rho+\epsilon H=(\rho+\epsilon H)^\intercal\geq 0$. Note that  $\sigma$ is a CS state,
  $\range{\sigma}\subseteq\range{\rho}$, and $\sigma$ is not proportional to $\rho$.  Hence, there exists a large enough
  $\epsilon_1>0$ such that $\alpha=\rho-\epsilon_1\sigma\geq 0$ is a CS state of rank at most five.  We have proven that
  $\alpha $ is separable, so $\alpha$ is a convex sum of real pure product states.  Hence $\range{\rho}$ has a real
  pure product states, and $\rho$ is separable.
\end{proof}
For the arbitrary multipartite system, we have the following sufficient and necessary condition. We omit the proof as it
is similar to the previous ones.
\begin{lem}
  \label{le:lastelm-multi}
  Suppose $\rho$ is CS state of rank six in the multipartite system. Then $\rho$ is separable if and only if the range
  of $\rho$ has a real product vector.
\end{lem}
Using Lemma \ref{le:lastelm}, one can show that for $2n$-partite $\rho$ Lemma \ref{le:lastelm-multi} still holds if
the condition  ``real'' is removed. However, further exploration needs to know whether it still holds when $\rho$ is a
$(2n+1)$-partite CS state. 

\section{Applications}
\label{sec:app}
In this section, we investigate the application of our results to a few widely useful states in quantum information. In
Sec.~\ref{subsec:sym}, we show that all multipartite PPT symmetric states of rank at most $4$ are separable.  We also
construct a sufficient condition by which an arbitrary multiqubit symmetric state is separable. In
Sec.~\ref{subsec:edge}, we show that the CS state in~\eqref{eq:4x4entangled} is an extreme and edge PP entangled
state. The property is unique for the above CS state, as we also construct a non-extreme and non-edge $4\times4$ PPT
entangled states of rank six, and it is invariant under partial transpose. Further in Sec.~\ref{subsec:nonnegative}, we
show that the entangled state in~\eqref{eq:4x4entangled} become nonnegative states by choosing suitable coefficients
$\lambda_j$'s. We investigate its geometric measure of entanglement, and its asymptotic version. In
Sec. \ref{subsec:hankel}, we highlight the relation between  the symmetric states and the Hankel and Toeplitz matrices.
Our results show that the CS states are intimately 
connected with the fundamentals of quantum information and matrix theory.

\subsection{Symmetric states}
\label{subsec:sym}

Note that CS states are a subclass of symmetric states. In this subsection, we will study the properties of the
symmetric stats. 
By \cref{lem:segre}, we have the following result.
\begin{corollary}
  \label{le:pptrank4=sep}
  All multipartite PPT symmetric states of rank at most $4$ are separable.
\end{corollary}
\begin{proof}
  
  Suppose $\rho$ is a separable state of rank at most $4$.  By the results in Ref.~\cite{chen2013separability}, $\rho$
  is separable except $\rho$ is supported in $3\otimes 3$ or $2\otimes 2\otimes 2$.  For the former case, by
  \cref{lem:segre}, $\rho$ contains a product vector in its range, hence separable. For the later case, it is separable
  according to the results in Ref.~\cite{eckert2002k}.
\end{proof}
From the discussion of \cref{thm:dqubit}, the results can be generalized as
\begin{corollary}
  Suppose $\rho$ is symmetric state and $\rho =\rho^{\intercal_1}$ in multi-qubit system, then $\rho$ is separable,
  where $\intercal_1$ is the partial transpose to the first subsystem.
\end{corollary}
Follows \cref{thm:bi=ful}, for symmetric states, we also have the following result.
\begin{corollary}
  The multipartite symmetric states are separable if and only if bi-separable.
\end{corollary}
Apply the similar discussion in \cref{thm:mul-n+1} for symmetric states with substituting real symmetric vector with
complex symmetric vector, we can prove the rank-$N$ and rank-$(N+1)$ symmetric states are separable.
\begin{corollary}
  Any multipartite symmetric states supported in the $N^{\otimes d}$ space of rank at most  $N+1$ are separable.
\end{corollary}

\subsection{Edge and extreme PPT entangled states}
\label{subsec:edge}
The edge and extreme PPT entangled states help understand the structure of PPT entangled states and separability
problem. Thus their construction is an important problem in quantum information. So far, the progress towards the problem
is little due to mathematical difficulty.

It follows from Lemma \ref{le:lastelm} that the range of the $4\times4$ PPT entangled state $\rho$ of rank six has no
product vector. So it is an edge PPT entangled state. An example is the state in~\eqref{eq:4x4entangled}.

We further claim that every $\rho$ is an extreme PPT entangled state~\cite{Chen_2013properties}. We prove the claim by
contradiction. Suppose some $\rho$ is not extreme, so $\rho=\alpha+\beta$, where $\alpha$ and $\beta$ are both PPT
entangled states, and neither of them is proportional to $\rho$. Thence$\range{\alpha},\range{\beta}\subseteq\range{\rho}$. We may choose a large enough $x>0$ such that $\sigma:=\rho-x\alpha$
is a PPT state of rank at most five. Since $\range{\sigma}\subseteq\range{\rho}$, the former does not have any product
vector. Therefore,  $\sigma$ is a PPT entangled state. The only possibility is that, $\sigma$ is a two-qutrit state of rank
four. It is a contradiction with Lemma~\ref{le:pptrank4=sep}. Consequently,  $\rho$ is an extreme PPT entangled
state. We thus  conclude the above findings as follows.

\begin{lem}
  \label{le:4x4edgeexe}
  Every $4\otimes4$ CS entangled state of rank six is an edge and extreme state.
\end{lem}
We emphasize that the above fact is a unique property for CS entangled states of rank six. Indeed, there exist
non-extreme and non-edge $4\times4$ PPT entangled states of rank six, which is invariant under partial transpose.  We
construct an example as follows. Let $\alpha=[C_0,C_1,C_2]^\dag [C_0,C_1,C_2]$ where
\begin{equation}
  \label{ea:Blokovi}
  C_0= \begin{bsmallmatrix}
                 0 & a & b \\
                 0 & 0 & 1 \\
                 0 & 0 & 0 \\
                 0 & 0 & 0
      \end{bsmallmatrix},
  C_1= \begin{bsmallmatrix}
                 0 & 0 & 0 \\
                 0 & 0 & c \\
                 0 & 0 & 1 \\
                 1 & 0 & -1/d
        \end{bsmallmatrix},
    C_2= \begin{bsmallmatrix}
                 0 & -1/b & 0  \\
                 0 & 1 & 0  \\
                 1 & -c & 0 \\
                 d & 0 & 0
        \end{bsmallmatrix},
\end{equation}
where  $a\in\real,a\neq 0$  and $b,c,d>0$. It has been shown in Theorem 23 of Ref.~\cite{chen2011description}  that
$\alpha$ is a two-qutrit PPT entangled state of rank four, and $\alpha^{\intercal_1}=\alpha$. Let 
$\beta=P\alpha P^\dag$ where $P=I_3\otimes(\ket{4}\bra{1}+\ket{2}\bra{2}+\ket{3}\bra{3})$.  So
$\rho=\alpha+\epsilon\beta$ with small enough $\epsilon>0$ is a $4\times4$ PPT entangled state of rank six, and
$\rho^{\intercal_1}=\rho$. By definition, $\rho$ is not extreme.

One can further show there exist product vectors in the range of
$\rho$. The six different row vectors of $[C_0,C_1,C_2]$ and $[C_0,C_1,C_2]P^\dag$ are as follows.
\begin{equation}
  \label{eq:rank6edge-0}
  \left\lbrace\quad
    \begin{aligned}
    \gamma_1&=a\bra{1,2}+b\bra{1,3}-\frac 1b\bra{3,2},\\
    \gamma_2&=\bra{1,3}+c\bra{2,3}+\bra{3,2},\\
    \gamma_3&=\bra{2,3}+\bra{3,1}-c\bra{3,2},\\
    \gamma_4&=\bra{2,1}-\frac 1d\bra{2,3}+d\bra{3,1},\\
    \gamma_5&=\bra{2,3}+\bra{3,4}-c\bra{3,2},\\
    \gamma_6&=\bra{2,4}-\frac 1d\bra{2,3}+d\bra{3,4}.\\
  \end{aligned}
  \right.
\end{equation}
It is equivalent to determining whether there exist product vectors in the space
$\mathit{span}\{\gamma_1,\gamma_2,\cdots,\gamma_6\}$. One can verify
\begin{equation}
  \label{eq:rank6edge-2}
  \gamma_3-\gamma_5+\gamma_4-\gamma_6=\big(\bra{2}+(1+d)\bra{3}\big)\big(\bra{1}-\bra{4}\big).
\end{equation}
Therefore, the range of $\rho$ has product vectors. By definition, $\rho$ is not edge.

\subsection{Nonnegative states}
\label{subsec:nonnegative}
The multipartite quantum state $\rho$ is nonnegative if all entries in its density matrix are real and nonnegative. Many
well-known states are indeed nonnegative states. Further, nonnegative states satisfy the additivity property of the
geometric measure of entanglement (GME), which is an extensively useful multipartite entanglement measure. It is defined
as
\begin{equation}
  \label{eq:gme}
  G(\rho):=
  -\log_2
  \max_{a_1,\ldots,a_d}
  \bra{a_1,\ldots,a_d}	
  \rho
  \ket{a_1,\ldots,a_d},
\end{equation}
where $\ket{a_1,\ldots, a_d}$ is a normalized product state. Since $\rho$ is nonnegative, it is known that
$\ket{a_1,\ldots,a_d}$ can be chosen as nonnegative states too.  See more details on page 12 of
Ref.\cite{zhu2010additivity}. Hence the asymptotic GME (AGME)
$$G^{\infty}(\rho)=\lim_{n\rightarrow\infty} {1\over  n}G(\rho^{\otimes n})$$
is equal to $G(\rho)$.

In our case, the entangled CS state $\rho$ in \eqref{eq:4x4entangled} is nonnegative if we choose $\lambda_5=\lambda_6$
and small enough $\lambda>0$. So it satisfies the facts in the last paragraph. It constructs a novel example of
entangled states whose GME is additive. In the following we compute the GME of the bipartite state $\rho$ in
Eq.~\eqref{eq:4x4entangled}, and thus its AGME.

It follows from Ref.~\cite{zhu2010additivity} that we can choose  $\ket{a_1}=\cdots=\ket{a_n}$ whose entries are
all real and nonnegative. Using the Lagrange multiplier, for the case $d=2$, the solution of the optimization problem
\begin{equation}
  \label{eq:opti}
  \max_{\norm{a}=1}
  \bra{a,a}	
  \rho
  \ket{a,a}
\end{equation}
satisfies the KKT condition
\begin{equation}
  \label{eq:KKT}
  \bra{a}\rho \ket{a,a} = \mu \ket{a}, \quad 
  a\in\real^4,
  \quad
  \norm{a}=1.
\end{equation}
The Perron–Frobenius theorem on the largest Z-eigenvalue of nonnegative tensor in Ref.~\cite{berman1994nonnegative} says
that if all entries of $\ket{a}$ are positive in \cref{eq:KKT}, then the corresponding $\mu$ equals the maximal value of
the optimization problem~\eqref{eq:opti}.

In particular, if $\lambda_i$ satisfy the following condition,
\begin{equation}
  \label{multiSsepv7(2):eq:2}
  \begin{split}
    \lambda_0 & = \lambda_3 + 3040 \lambda_5 -\frac{11000}{81}\lambda_7,\\
    \lambda_1 & = \lambda_3 + 2016 \lambda_5,\\
    \lambda_2 & = \lambda_3 + 1054 \lambda_5 - \frac{11000}{81} \lambda_7,
  \end{split}
\end{equation}
then $a = \frac{1}{2}(1,1,1,1)^{\intercal}$ and
\begin{math}
  \mu = \frac{\lambda_0}{4} + 16 \lambda_4 + 248 \lambda_5 + \frac{250}{27}\lambda_7
\end{math}
satisfy the condition~\eqref{eq:KKT}. Hence, we have
\begin{equation}
  \label{multiSsepv7(2):eq:4}
  G(\rho)  
  =
  G^{\infty}(\rho)
  = -\log_2( \frac{\lambda_0}{4} + 16 \lambda_4 + 248 \lambda_5 + \frac{250}{27}\lambda_7).
\end{equation}

\subsection{Connection to Hankel and Toeplitz matrices}
\label{subsec:hankel}
In this subsection, we will consider some classes of symmetric states. In the multi-qubit system, the symmetric states
can be represented by the Dicke basis
\begin{equation}
  \label{multiSsepv7:eq:2}
  \rho = \sum_{i,j=0}^dm_{ij}\ket{\ddk[i]}\bra{\ddk[j]}.
\end{equation}
Denote by $M_{\rho}$ the $(d+1)\times (d+1)$ matrix whose $(i,j)$-th entry is $m_{ij}$. Hence, the corresponding
semidefinite positive matrix $M_{\rho}$ can be used to represent the symmetric states which has $\frac{(d+1)(d+2)}{2}$ degree of
freedom.

In particular, $\rho$ is a DS state if $M_{\rho}$ is a semidefinite positive diagonal matrix . Moreover, $\rho$ is a CS state, if
$M_{\rho}$ is a Hankel matrix, which is
\begin{equation}
  \label{multiSsepv7:eq:3}
  M_{\rho} =\hskip-10pt
  \vcenter{
    \begin{tikzpicture}
      \matrix(m)[matrix of math nodes,left delimiter={[},right delimiter={]},row sep=0.5em, column sep=0.1em] {
        a_0    & a_1    &a_2       &   &a_d\\
        a_1    & a_2    &         &   & a_{d-1} \\
        a_2    &       &          &         &     \\
        & &          &         &   \\
        a_d    & a_{d-1}&    &         &a_{2d}\\
      }; \path[dots] (m-5-1) edge (m-1-5) (m-5-2) edge (m-2-5) (m-1-3) edge (m-1-5) (m-3-1) edge (m-5-1) (m-5-2) edge
      (m-5-5) (m-5-5) edge (m-2-5);
    \end{tikzpicture}
  }\hskip-118pt.
\end{equation}
Equivalently,
\begin{equation}
  \label{multiSsepv7:eq:4}
  \rho = \sum_{k=0}^{d}a_k \left( \sum_{i+j=k}\ket{\ddk[i]}\bra{\ddk[j]}\right).
\end{equation}

By our results, if $M_{\rho}$ is semidefinite positive, then $\rho$ is S-separable. Hence, we have the following result about the
Hankel matrix.
\begin{lem}
  Suppose $M$ is an $N\times N $ Hankel matrix, then it is semidefinite positive if and only if has the following
  decomposition
  \begin{equation}
    \label{multiSsepv7:eq:6}
    M = \sum_{i=1}^{r}\lambda_i z_i z_i^\intercal,
  \end{equation}
  where $r = \rank{M}$, $\lambda_i>0$, and
  \begin{equation}
    \label{multiSsepv7:eq:10}
    z_i =
    \begin{bmatrix}
      1\\ t_i \\ t_i^2\\ \vdots \\t_i^{N-1}
    \end{bmatrix}
    \text{ or }
    \begin{bmatrix}
      0\\0\\ 0\\\vdots \\0\\1
    \end{bmatrix}, t_i\in\real.
  \end{equation}
\end{lem}
\begin{proof}
  Let $V$ be a diagonal matrix $V = \mathrm{diag}(\sqrt{C_{d-1}^0},\sqrt{C_{d-1}^1},\ldots,\sqrt{C_{d-1}^{d-1}})$, then
  $VMV^{\intercal}$ is also a Hankel matrix.  If $M$ is semidefinite positive, $VMV^{\intercal}$ is also semidefinite
  positive. We can construct a 
  (N-1)-qubit CS state $\rho$ according to \cref{thm:dqubit}, $\rho$ is separable.  must be S-separable. That is, it can
  be written as
  \begin{equation}
    \label{multiSsepv7:eq:8}
    \rho = \sum_{i=1}^{r}\lambda_i\ket{x_i}^{\otimes N-1}\bra{x_i}^{\otimes N-1},x\in\real^2,
  \end{equation}
  where $r = \rank{\rho}$.
  If $\ket{x_i}\not\propto \ket{0}$, we can assume $\ket{x_i} = \ket{1}+t_i\ket{0}$. Then
  \begin{equation}
    \label{multiSsepv7:eq:9}
    {(\ket{0}+t_i\ket{1})}^{\otimes (N-1)} = \sum_{j=0}^{N-1}t_i^{j}\sqrt{C_{N-1}^j}\ket{D_{N-1,j}}.
  \end{equation}
  If $\ket{x_i}=\ket{0}$, then $\ket{x_i}^{\otimes (N-1)} =\ket{D_{N-1,N-1}}$. Therefore,
  \begin{equation}
    \label{multiSsepv7:eq:11}
    VMV^{\intercal} = \sum_i Vz_iz_i^{\intercal}V^{\intercal},
  \end{equation}
  where $z_i$ has the form in Eq.~\eqref{multiSsepv7:eq:10}. The proof completes by cancel $V$ on the above equation.
\end{proof}

It is also of interest to consider some other classes of symmetric states.  In particularly, if $M_{\rho}$ is a Toeplitz matrix,
\begin{equation}
  \label{multiSsepv7:eq:5}
  M_{\rho} =\hskip-10pt \vcenter{
    \begin{tikzpicture}
      \matrix(m)[matrix of math nodes,left delimiter={[},right delimiter={]},row sep=0.5em, column sep=0.1em] {
        a_0    & a_1    &       &   &a_d\\
        a_1    &     &         &   &  \\
        &       &          &         &     \\
        & &          &         &  a_1 \\
        a_d    &  &    &      a_1   &a_{0}\\
      }; \path[dots] (m-1-2) edge (m-1-5) (m-5-1) edge (m-5-4) (m-1-1) edge (m-5-5) (m-2-1) edge (m-5-4) (m-1-2) edge
      (m-4-5) (m-2-1) edge (m-5-1) (m-1-5) edge (m-4-5);
    \end{tikzpicture}
  }\hskip-138pt.
\end{equation}
Equivalently,
\begin{equation}
  \rho = \sum_{k=0}^{d}a_k \left( \sum_{|i-j|=k}\ket{\ddk[i]}\bra{\ddk[j]}\right).
\end{equation}
The Hankel matrix is closely related to the Toeplitz matrix, in fact, a Hankel matrix is an upside-down Toeplitz matrix. 
We conjecture that these states are seprable.
\begin{conjecture}
  \label{conj:Toeplitz}
  Suppose $\rho$ is multi-qubit symmetric states. If the corresponding $M_{\rho}$ is Toeplitz, then $\rho$ is separable.
\end{conjecture}
\section{Conclusion}%
\label{sec:con}
We have studied the property of completely symmetric (CS) states in the multipartite system.

First, we have proved that the CS states admit some nice features. For example, the CS state is separable if and only it
is bi-separable or S-separable. Moreover, it is reducible over real if and only the reduced states are reducible over
real for $d$-partite ($d>2$) system. In particular, the CS states are separable if and only if every real irreducible
components are separable, which makes the separability problem easier to solve.

Second, we distinguished   the separability of some low-rank CS states. It turns out that all the multi-qubit and
two-qutrit CS states are separable. Moreover, all the CS states of rank at
most $5$ or $N+1$ are separable both in the bipartite and multipartite system. As for the rank-$6$ or $N+2$ CS states,
they are separable if and only their range contains a real product vector. However, there exists a CS entangled
state of rank $6$. So the CS structure cannot guarantee the separability, and the further exploration is needed. 

Third, we showed some application of CS states  to some widely useful states in quantum information.
For example the symmetric states, edge and extreme states, nonnegative states. We also have suggested the connection between
the separability to the decomposition of  Hermitian matrices, such as the Hankel and Toplitz matrices.

There are many problems one may consider in the future study of entanglement.
\begin{enumerate}
\item  Is \cref{conj:Toeplitz} true?
\item  Are multi-qutrit CS states separable?\\
\item  What is the relation  between the separability and local ranks for CS states? Does the Conjecture~24 hold for the
CS states in Ref.~\cite{chen2013separability}?\\
\item  For a CS state, if all the reduced states are separable, can we conclude that the state is separable?\\
\item  How is to find an entanglement witness for a given entangled CS state?\\
\item  Can the symmetric extension criterion detect the separability of CS states completely? Is the corresponding
numerical algorithm practical? Can we find a variation of symmetric extension criterion to fit the structure of CS states? 
\item
The CS states correspond to the supersymmetric tensor in tensor analysis, and the tensor rank of two copies of special symmetric tensors has been considered recently \cite{Chen2018the}. Is it possible to compute the tensor rank of two copies of some CS states?
\end{enumerate}

\section*{Acknowledgments}

LC and YS were supported by the NNSF of China (Grant No. 11871089), and the Fundamental Research Funds for the Central
Universities (Grant Nos. KG12040501, ZG216S1810 and ZG226S18C1).
%


\appendix
\section{Real Product vectors  in $\range{\sigma}$ of ~\eqref{eq:sigma}}%
\label{append}
Let $\sigma$ is the CS state in \cref{eq:sigma}. 
Suppose the subspace $P$ is the range of $\sigma$, then it is spanned by the following vectors:
\begin{equation}
  \ket{\phi_i} = \ket{x_i,x_i},\text { for } i=0,1,\ldots, 6
\end{equation}
and
\begin{equation}
  \begin{split}
    \ket{x_i}& = \ket{i},i=0,1,2,3\\
    \ket{x_4}&= \ket{0}+\ket{1}+\ket{2}+\ket{3},\\
    \ket{x_5}&= \ket{0}+2\ket{1}+3\ket{2}+4\ket{3},\\
    \ket{x_6}&= \ket{0}-2\ket{1}+3\ket{2}-4\ket{3},\\
  \end{split}
\end{equation}
In this appendix, we want to find that $P$ contains exactly 8 real  product vectors. The extra real product vector
is$
\ket{\phi_7} = \ket{x_7,x_7}$, where
\begin{equation}
  \label{phi7}
  \ket{x_7} = \ket{0} - \frac{8}{3}\ket{1}+\ket{2}-\frac{8}{3}\ket{3}.
\end{equation}
\begin{proof}
  Suppose
  \begin{equation}
    \label{multiSsepv5:eq:7}
    v =
    \begin{bmatrix}
      a_0\\
      a_1\\a_2\\a_3
    \end{bmatrix}.
  \end{equation}
  If $\ket{w}=\ket{v,v}$ is a linear combination of $\ket{\phi_i},i=0,1,\ldots,6$, then the matrix 
  \begin{equation}
    \label{multiSsepv5:eq:8}
    \begin{bmatrix}
      \phi_0&\phi_1&\cdots &w
    \end{bmatrix}
  \end{equation}
  must be of rank $7$.
  Remove the repeated rows, we have,
  \begin{equation}
    \label{multiSsepv5:eq:9}
    \begin{bsmallmatrix}
      1&0&0&0&  1 & 1 & 1 &a_0^2\\
      0&1&0&0&  1 & 4 & 4 &a_1^2\\
      0&0&1&0&  1 & 9 & 9 & a_2^2\\
      0&0&0&1&  1 & 16 & 16 &a_3^2\\
      0&0&0&0&  1 &2 & -2 &a_0a_1\\
      0&0&0&0&  1 &3 & 3 & a_0a_2\\
      0&0&0&0&  1 &4 & -4 & a_0a_3\\
      0&0&0&0&  1 &6 & -6 & a_1a_{2}\\
      0&0&0&0&  1 &8 & 8 & a_1a_3\\
      0&0&0&0&  1 & 12 & -12 & a_2a_3\\
    \end{bsmallmatrix}
  \end{equation}
  is of rank $7$. It follows the submatrix
  \begin{equation}
    \label{multiSsepv5:eq:10}
    \begin{bsmallmatrix}
      1 &2 & -2 &a_0a_1\\
      1 &3 & 3 & a_0a_2\\
        1 &4 & -4 & a_0a_3\\
        1 &6 & -6 & a_1a_{2}\\
        1 &8 & 8 & a_1a_3\\
       1 & 12 & -12 & a_2a_3\\
    \end{bsmallmatrix}
  \end{equation}
  is of rank $3$.
  Which is equivalent to
  \begin{equation}
    \label{multiSsepv5:eq:11}
    \begin{bmatrix}
      1 & 5 & a_0a_2-a_0a_{1}\\
      2 & -2 & a_0a_3-a_{0}a_{1}\\
       4 & -4 & a_1a_{2}-a_{0}a_{1}\\
        6 & 10 & a_1a_3-a_{0}a_{1}\\
       10 & -10  & a_2a_3-a_{0}a_{1}\\
    \end{bmatrix}
  \end{equation}
  is of rank $2$. Furthermore,
  \begin{equation}
    \label{multiSsepv5:eq:12}
    \begin{bmatrix}
      -12 & a_0a_3+a_0a_1-2a_0a_2\\
      -24 & a_1a_2 + 3a_0a_1 - 4 a_0a_2\\
      -20 & a_1a_3 + 5a_0a_1 - 6 a_0a_2\\
      -60 & a_2a_3 + 9a_0a_1 - 10 a_0a_2\\
    \end{bmatrix}
  \end{equation}
  is of rank $1$. Hence
  \begin{equation}
    \label{multiSsepv5:eq:13}
    \begin{split}
      \frac{a_0a_3+a_0a_1-2a_0a_2}{3} & = \frac{a_1a_2 + 3a_0a_1 - 4 a_0a_2}{6}\\
      & = \frac{a_1a_3 + 5a_0a_1 - 6 a_0a_2}{5}\\
      & = \frac{a_2a_3 + 9a_0a_1 - 10 a_0a_2}{15}.
    \end{split}
  \end{equation}
  If $a_0\neq 0$, we can assume that $a_0=1$. Then \cref{multiSsepv5:eq:13} implies that
  \begin{equation}
    \label{multiSsepv5:eq:14}
    \begin{split}
      2a_3& = a_1a_2 + a_1,\\
      5a_3& = a_2a_3 + 4 a_1,\\
      3a_1a_3 & = a_2a_3 - 6a_1 + 8a_2.\\
    \end{split}
  \end{equation}
  If $a_1=0$, then $a_2=a_3=0$, where $\ket{w} = \ket{\phi_0}$.
  Solving \cref{multiSsepv5:eq:14}, we have
  \begin{equation}
    \label{multiSsepv5:eq:15}
    a_2 = 3 \text{ or } a _{ 2} =1.
  \end{equation}
  If $a_2=3$, by \cref{multiSsepv5:eq:14}, $a_1=\pm 2,a_3 = \pm 4$. Then $\ket{w} = \ket{\phi_5}$ or
  $\ket{\phi_6}$.

  If $a_2=1$, by \cref{multiSsepv5:eq:14}, we have $a_1=1$ or $-\frac{8}{3}$. If $a_1=1$, we have $a_3=1$, where
  $\ket{w}=\ket{\phi_{4}}$. If $a_1= - \frac{8}{3}$, then $a_3= a_1$, which coincides with \cref{phi7}.

  If $a_0=0,a_1\neq 0$, we can assume that $a_1=1$, by \cref{multiSsepv5:eq:13}, we have
  \begin{equation}
    \label{multiSsepv5:eq:16}
    0 = \frac{a_2}{6}= \frac{a_3}{5}.
  \end{equation}
  It follows that $\ket{w} = \ket{\phi_1}$.
  It remains to consider the case $a_0=a_1=0$, we assume $a_2=1$, otherwise, $\ket{\phi_{7}}=\ket{\phi_3}$.
  By \cref{multiSsepv5:eq:13}, we have $a_3=0$, i.e., $\ket{w}=\ket{\phi_2}$.

  Above all, the subspace $P$ contains exactly $8$ real product vectors, that is
  $\ket{\phi_0},\ket{\phi_1},\ldots,\ket{\phi_7}$.
\end{proof}
\end{document}